\documentclass[12pt,draftcls,onecolumn]{IEEEtran}
\usepackage{cite}
\usepackage{amsmath,amssymb,amsthm}    % need for subequations
\usepackage[dvips]{graphicx}   % need for figures
\usepackage{verbatim}   % useful for program listings
\usepackage{color}      % use if color is used in text
\usepackage{subfigure}  % use for side-by-side figures
\usepackage{epsfig}
\usepackage{float}%\usepackage{keywords}
\usepackage{algorithm}
\usepackage{algorithmic}
\usepackage{setspace}
\usepackage{subfigure}
\newtheorem{lemma}{\bf Lemma}
\newtheorem{theorem}{\bf Theorem}

\newtheorem{assumption}{\bf Assumption}

\makeatletter
\newcommand*{\rom}[1]{\expandafter\@slowromancap\romannumeral #1@}
\makeatother

\begin{document}
\title{\textbf{Distributed Bayesian Detection with Byzantine Data}}
\author{Bhavya~Kailkhura,~\IEEEmembership{Student Member,~IEEE}, Yunghsiang~S. Han,~\IEEEmembership{Fellow,~IEEE},
Swastik~Brahma,~\IEEEmembership{Member,~IEEE}, Pramod~K.~Varshney,~\IEEEmembership{Fellow,~IEEE}
\thanks{B. Kailkhura, S. Brahma and P. K. Varshney are with Department of EECS, Syracuse University, Syracuse, NY 13244. (email: bkailkhu@syr.edu; skbrahma@syr.edu; varshney@syr.edu)}
\thanks{Y. S. Han is with EE Department, National Taiwan University of Science and Technology, Taiwan, R. O. C. (email: yshan@mail.ntust.edu.tw)}}
\date{}
\maketitle
\begin{abstract} 
In this paper, we consider the problem of distributed Bayesian detection in the presence of Byzantines in the network. 
It is assumed that a fraction of the
nodes in the network are compromised and reprogrammed by an adversary to transmit
false information to the fusion center (FC) to degrade detection performance. 
The problem of distributed detection is formulated as a binary hypothesis test at the FC based on 1-bit data sent by the sensors.
The expression for minimum attacking power required by the Byzantines to blind the FC is obtained. More specifically, we show that
above a certain fraction of Byzantine attackers in the network, the detection scheme becomes completely incapable of utilizing the sensor data for detection.
We analyze the problem under different attacking scenarios and derive results for different non-asymptotic cases. It is found that existing asymptotics-based results do not
hold under several non-asymptotic scenarios.
When the fraction of Byzantines is not sufficient to blind the FC, we also provide closed form expressions for the optimal attacking strategies for the Byzantines that most degrade the detection performance.
\end{abstract}

\section{Introduction}
Distributed detection is a well studied topic in the detection
theory literature~\cite{Varshney, Viswanathan, veer}. 
In distributed detection systems, due to bandwidth and energy constraints, the nodes often make a 1-bit local decision regarding the presence of a phenomenon before sending it to the fusion center (FC). Based on the local decisions transmitted by the nodes, the FC makes a global decision about the presence of the phenomenon of interest. Distributed detection was originally motivated by its applications in
military surveillance but is now being employed in a wide variety of applications such as distributed spectrum sensing (DSS) using cognitive radio networks (CRNs) and traffic and environment monitoring.

In many applications, a large number of inexpensive and less reliable nodes that can provide dense coverage are used
to provide a balance between cost and functionality.
The performance of such systems strongly depends on
the reliability of the nodes in the network. The robustness of
distributed detection systems against attacks is of utmost importance. The
distributed nature of such systems makes them quite vulnerable
to different types of attacks. 
In recent years, security issues of such
distributed networks are increasingly being studied within the networking~\cite{a4}, signal processing~\cite{a3} and information theory communities~\cite{a5}. One typical attack on such networks is a Byzantine attack. While Byzantine attacks (originally proposed by \cite{Lamport}) may, in general, refer to many types of malicious behavior, our focus in this paper is on data-falsification attacks ~\cite{avemp,frag, Rifa, Marano, Rawat, Kailkhura2013, Kailkhura, aditya}. In this type of attack, an attacker may send false (erroneous) data to the FC to degrade detection performance. In this paper, we refer to such a data falsification attacker as a Byzantine and the data thus generated is referred to as Byzantine data.

We formulate the signal detection problem as a binary hypothesis testing problem with the two hypotheses $H_{0}$ (signal is absent) and $H_{1}$ (signal is present).
We make the conditional i.i.d. assumption under which observations at the nodes are conditionally independent and identically
distributed given the hypothesis.
We assume that the FC is not compromised,
and is able to collect data from all the nodes in the network via error free communication channels.\footnote{In this work, we do not consider  how individual nodes deliver their data to the fusion center except that the Byzantines are not able to alter the transmissions of honest nodes.}   
We also assume that the FC does not know which node
is Byzantine, but it knows the fraction of Byzantines in the network.\footnote{In practice, the fraction of Byzantines in the network can be learned by observing the data sent by the nodes at the FC over a time window; however, this study is beyond the scope of this work.} We consider the problem of distributed Bayesian detection with prior probabilities of hypotheses known to both the FC and the attacker. The FC aims to minimize the probability of error by choosing the optimal fusion rule.
\subsection{Related Work}
Although distributed detection has been a very active field of research in the past, security problems in distributed detection
networks gained attention only very recently. In \cite{Marano}, the authors considered the problem of distributed detection in the presence of Byzantines under the Neyman-Pearson (NP) setup and determined the optimal attacking strategy which minimizes the detection error exponent. This approach based on Kullback-Leibler divergence (KLD) is analytically tractable and yields approximate results in non-asymptotic cases. They also assumed that the Byzantines know the true hypothesis, which obviously is not satisfied in practice but does provide a bound. 
In \cite{Rawat}, the authors analyzed the same problem in the context of collaborative spectrum sensing under Byzantine Attacks. They relaxed the  assumption of perfect knowledge of the hypotheses by assuming that the Byzantines
determine the knowledge about the true hypotheses from their own sensing observations. A variant of the above formulation was explored in \cite{Kailkhura2013,bhavyaj}, where the authors addressed the problem of optimal Byzantine attacks (data falsification) on
distributed detection for a tree-based topology and extended the results of ~\cite{Rawat} for tree topologies. By assuming
that the cost of compromising nodes at different levels of the tree is different, they found the optimal Byzantine
strategy that minimizes the cost of attacking a given tree. Schemes for Byzantine node identification have
been proposed in ~\cite{aditya,Rawat,a1,a2}.
Our focus is considerably different from Byzantine node identification schemes in that we do not try to
authenticate the data; we consider most effective attacking
strategies and distributed detection schemes that are robust against attacks.

\subsection{Main Contributions}
All the approaches discussed so far consider distributed detection under the Neyman-Pearson (NP) setup. In this paper, we consider the distributed Bayesian detection problems with known prior probabilities of hypotheses. We assume that the Byzantines do not have perfect knowledge about the true state of the phenomenon of interest. In addition, we also assume
that the Byzantines neither have the knowledge nor control over the thresholds used to make local decisions at the nodes. Also, the probability of detection and the probability of false alarm of a node are assumed to be the same for every node irrespective of whether they are honest or Byzantines.
In our earlier work \cite{bhavyaarxiv} on this problem,
we analyzed the problem in the asymptotic regime. Adopting Chernoff information as our performance metric, we studied the performance of a distributed detection system with Byzantines in the asymptotic regime. We summarize our results in the following theorem.
\begin{theorem}[\cite{bhavyaarxiv}]
\label{th1}
Optimal attacking strategies, $(P_{1,0}^*, P_{0,1}^*)$, which minimize the Chernoff information are
\[ (P_{1,0}^*, P_{0,1}^*)  \left\{ \begin{array}{rll}
				(p_{1,0}, p_{0,1})  & \mbox{if}\ \alpha \geq 0.5 \\
			   (1,1)  & \mbox{if}\ \alpha < 0.5
				\end{array}\right. ,
\] 
where, $(p_{1,0}, p_{0,1})$ satisfy $\alpha(p_{1,0}+p_{0,1})=1$.
\end{theorem}
 \begin{table}
 \begin{center}
 \caption{Different scenarios based on the knowledge of the opponent's strategies}
 \renewcommand{\arraystretch}{2}
 \resizebox{10cm}{!} {
   \begin{tabular}{ c | c | c }
     \hline
       Cases  & Attacker has the knowledge of the FC's strategies  & FC has the knowledge of Attacker's strategies \\ \hline
      Case 1 & No & No \\
      Case 2 & Yes & No \\
      Case 3 & Yes & Yes \\
      Case 4 & No & Yes \\
      \hline
   \end{tabular}
 }
 \label{table10}
 \end{center}
 \end{table}

In our current work,
we significantly extend our previous work and focus on a \textit{non-asymptotic} analysis for the Byzantine attacks on distributed Bayesian detection.
First, we show that above a certain fraction of Byzantines in the network, the data fusion scheme becomes completely incapable (blind) and it is not possible to design a decision rule at the FC that can perform better than the decision rule based just on prior information. We find the minimum fraction of Byzantines that can blind the FC  and refer to it as the \textit{critical power}. Next, we explore the optimal attacking strategies for the Byzantines under different scenarios.
In practice, the FC and the Byzantines will optimize their utility by choosing their actions based on the knowledge of their opponent's behavior. This motivates us to address the question: what are the optimal attacking/defense strategies given the knowledge of the opponent's strategies? Study of these practically
motivated questions requires non asymptotic analysis, which is systematically studied in this work.
By assuming the error probability to be our performance metric, we analyze the problem in the non asymptotic
regime. Observe that, the probability of error is a function of the fusion rule, which is under the control
of the FC. This
gives us an additional degree of freedom to analyze the Byzantine attack under different practical scenarios
where the FC and the Byzantines may or may not have knowledge of their opponent's strategies (For a description of different scenarios see Table~\ref{table10}). 
It is found that results based on asymptotics do not hold under several non-asymptotic scenarios.
More specifically, when the FC does not have knowledge of attacker's strategies, results for the non-asymptotic case are different from those for the asymptotic case. However, if the FC has complete knowledge of the attacker's strategies and uses the optimal fusion rule to make the global decision, results obtained for this case are the same as those for the asymptotic case. 
Knowledge of the behavior of the attacker in the non-asymptotic regime enables the analysis of many related questions, such as the design of the optimal detector (fusion rule) and effects of strategic interaction between the FC and the attacker. In the process of analyzing the scenario where the FC has complete knowledge of its opponent's strategies, we obtain a closed form expression of the optimal fusion rule.
To summarize, our main contributions are threefold. 
\begin{itemize}
\item In contrast to previous works, we study the problem of distributed detection with Byzantine data in the Bayesian framework.
\item We analyze the problem under different attacking scenarios and derive
closed form expressions for optimal attacking strategies for different non-asymptotic cases.
\item In the process of analyzing the scenario where the FC has complete knowledge of its
opponent's strategies, we obtain a closed form expression for the optimal fusion rule.
\end{itemize}

The signal processing problem considered in this paper is
closest to~\cite{Rawat}. The approach in~\cite{Rawat}, based on Kullback-Leibler divergence (KLD), is analytically tractable and yields approximate results in non-asymptotic cases.
Our results, however, are not a direct application of those
of~\cite{Rawat}. While as in~\cite{Rawat} we
are also interested in the optimal attack strategies, our objective function and, therefore, techniques of finding them are different. 
In contrast to~\cite{Rawat}, where only optimal strategies to blind the FC were obtained, we also provide closed form expressions for the optimal attacking strategies for the Byzantines that most degrade the detection performance when the fraction of Byzantines is not sufficient to blind the FC. 
In fact, finding the optimal Byzantine attacking strategies is only the first step toward designing a robust distributed detection system. Knowledge of these attacking strategies can be used to implement the optimal detector at the FC or to implement an efficient reputation based identification scheme~\cite{Rawat,covert} ( thresholds in these schemes are generally a function of attack strategies).
Also, the optimal attacking distributions in certain cases have the minimax property and, therefore, the knowledge of these optimal attack strategies can be used to implement the robust detector.

The rest of the paper is organized as follows.
Section~\ref{sec2} introduces our system model, including the Byzantine attack model.
In Section~\ref{sec3}, we provide the closed form expression for the critical power above which the FC becomes blind.
Next, we discuss our results based on non-asymptotic analysis of the distributed Bayesian detection system with Byzantine data for different scenarios.
In Section~\ref{sec5}, we analyze the problem when Byzantines do not have any knowledge about the fusion rule used at the FC.
Section~\ref{majority} discusses the scenario where Byzantines have the knowledge about the fusion rule used at the FC, but the FC does not know the attacker's strategies.
Next in Section~\ref{sec6}, we extend our analysis to the scenario where both the FC and the attacker have the knowledge of their opponent's strategies and act strategically to optimize their utilities.
Finally, Section~\ref{sec7} concludes the paper.

\begin{figure}[t!]
  \centering
    \includegraphics[height=0.25\textheight, width=0.4\textwidth]{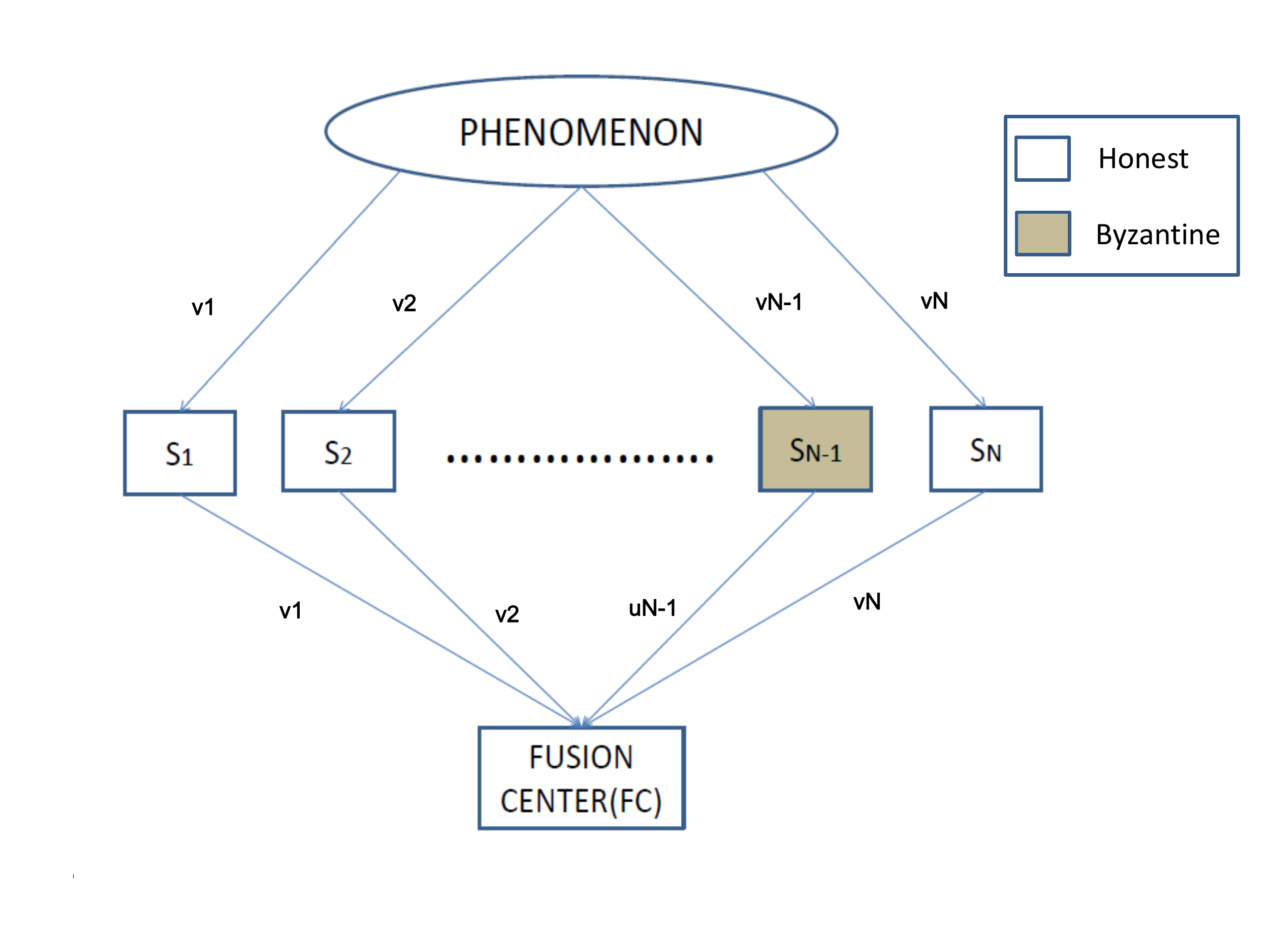}
    \caption{System Model}
    \label{model}
\end{figure}

\section{Distributed detection in the presence of Byzantines}
\label{sec2}
%\label{sec2}
%\begin{figure}[t!]
%  \centering
%    \includegraphics[height=0.25\textheight, width=0.4\textwidth]{sys.pdf}
%    \caption{System Model}
%    \label{model}
%\end{figure}
Consider two hypotheses %testing problem with the two hypotheses
$H_{0}$ (signal is absent)
and $H_{1}$ (signal is present).
Also, consider a parallel network (see Figure~\ref{model}), comprised of a central entity (known as the Fusion Center (FC))
and a set of $N$ sensors (nodes), which
faces the task of determining which of the two hypotheses is true.
Prior probabilities of the two hypotheses $H_{0}$ and $H_{1}$
are denoted by $P_{0}$ and $P_{1}$, respectively. 
The sensors observe %signals emitting from
the phenomenon, carry out local computations to decide the presence or
absence of the phenomenon, % as a 0-1 decision variable,
and then send their local decisions to the FC that yields a final decision
after processing the local decisions.
Observations at the nodes are assumed to be
conditionally independent and identically distributed given the hypothesis.
A Byzantine attack on such a system compromises some of the nodes
which may then intentionally send falsified local decisions to the FC to make the final
decision incorrect.
We assume that a fraction $\alpha$ of the $N$ nodes which observe the phenomenon
%which %the $N$ nodes
have been compromised by an attacker.
We consider the communication channels to be error-free.
Next, we describe the modus operandi %functionalities
of the sensors and the FC in detail.

\subsection{Modus Operandi of the Nodes}
\label{attack model}
Based on the observations,
%based on which
each node $i$ %at level $k$
%acts as a source in that it
makes a one-bit local decision $v_{i} \in \{ 0,1\}$ regarding the absence or presence of the phenomenon
using the likelihood ratio test
    \begin{equation}
        \label{eqn1}
        \dfrac{p_{Yi}^{(1)}(y_{i})}{p_{Yi}^{(0)}(y_{i})} \quad \mathop{\stackrel{v_i = 1}{\gtrless}}_{v_i = 0} \quad \lambda,
    \end{equation}
    where $\lambda$ 
    is the identical threshold\footnote{It has been shown that the use of identical thresholds 
    is asymptotically optimal~\cite{tsit}.}
    %\footnote{It has been shown that the use of identical thresholds is asymptotically optimal \cite{tsit}.}
    used at all the sensors and $p_{Yi}^{(k)} (y_{i})$ is the conditional probability density function (PDF) of
    observation $y_i$ under the hypothesis $H_k$.
%
%After making its one-bit local decision $v_{k,i}$, node $i$ sends $u_{k,i}$ (which may not be the same as $v_{k,i}$) to its parent node at level $k-1$.
%It also receives the decisions $u_{k',j}$ of all successors $j$
%at levels $k' \in [k+1,L]$, which are forwarded to node $i$ by its immediate children, and
%%It also receives decisions $u_{k',j,k+1}$ relayed by its children at Level $k+1$,
%%where $k'\in\{k+1,...,K\}$ and $j\in C_{k,i}$ and
%%act as a relay by forwarding
%forwards them to its parent node at level $k-1$.
%We also assume error-free communication channels between children and the parent nodes.
%After making their local decisions, the nodes send their local de
%After making their local decisions, the nodes send inputs
%to a central entity (referred to as the Fusion Center (FC))
%which takes a global decision regarding the absence or presence of the phenomenon
%after processing the inputs.
Each node $i$, after making its one-bit local decision $v_{i}$,
%Specifically, we consider that the input
%node $i$
sends $u_i\in\{0,1\}$ to the FC, where $u_i = v_i$
if $i$ is an uncompromised (honest) node, but for a compromised (Byzantine) node $i$,
$u_i$ need not be equal to $v_i$.
%an uncompromised (honest) node %(referred to as honest sensor)
%sends its true decision $v_i$ to the FC, % to the FC,
%%a central entity (referred to as the Fusion Center (FC)) for processing
%while a compromised (Byzantine) node %$i$ (referred to as Byzantine sensor)
%sends $u_{i}$ (which may not be the same as $v_{i}$) to the FC
%to make the global decision incorrect.
%A Byzantine node, in order to undermine the network performance, may alter its decision
%prior to transmission.
We denote the probabilities of detection and false alarm of each
node $i$ in the network by $P_{d}=P(v_{i}=1|H_{1})$
and $P_{f}=P(v_{i}=1|H_{0})$, respectively,
which hold for both
%irrespective of whether it is
uncompromised nodes as well as
%or has been
compromised nodes.
%which are assumed
%to be the same for
%every node irrespective of whether they are
%%both all the
%honest or Byzantine nodes.
In this paper, we assume that each Byzantine decides to attack independently
relying on its own observation and decision regarding
the presence of the phenomenon.
Specifically, we define the following strategies $P_{j,1}^H$, $P_{j,0}^H$ and $P_{j,1}^B$, $P_{j,0}^B$ ($j \in \{0,1\}$)
for the honest and Byzantine nodes, respectively:\\
\vspace{-0.04in}
Honest nodes:
\begin{equation}
\label{honest}
P_{1,1}^H=1-P_{0,1}^H=P^{H}(x=1|y=1)=1
\end{equation}
\begin{equation}
P_{1,0}^H=1-P_{0,0}^H=P^{H}(x=1|y=0)=0
\end{equation}
				
\noindent
Byzantine nodes:
\begin{equation}
P_{1,1}^B=1-P_{0,1}^B=P^{B}(x=1|y=1)
\end{equation}
\begin{equation}
P_{1,0}^B=1-P_{0,0}^B=P^{B}(x=1|y=0)
\end{equation}
$P^{H}(x=a|y=b)$ ($P^{B}(x=a|y=b)$) is the probability that an honest (Byzantine) node sends $a$ to the FC
when its actual local decision is $b$. From now
onwards, we will refer to Byzantine flipping probabilities simply by $(P_{1,0}, P_{0,1})$. We also assume that the FC is not aware of the exact set of Byzantine nodes and considers each node $i$ to be Byzantine with a certain probability $\alpha$.

\subsection{Binary Hypothesis Testing at the Fusion Center}
\label{Testing}
We consider a Bayesian detection problem where the
performance criterion at the FC is the probability of error.
The FC receives decision vector, $\mathbf{u}=[u_1,\cdots,u_N]$, from the nodes and makes the global decision about the phenomenon by considering the maximum a \textit{posteriori} probability (MAP) rule which is given by
\begin{equation*}
P(H_1|\mathbf{u}) \quad \mathop{\stackrel{H_1}{\gtrless}}_{H_0} \quad  P(H_0|\mathbf{u})
\end{equation*}
 or equivalently,
\begin{equation*}
\dfrac{P(\mathbf{u}|H_1)}{P(\mathbf{u}|H_0)} \quad \mathop{\stackrel{H_1}{\gtrless}}_{H_0} \quad  \dfrac{P_0}{P_1}.
\end{equation*}
Since the $u_i$’s are independent of each other, the MAP rule simplifies to
a $K$-out-of-$N$ fusion rule~\cite{Varshney}.
The global false alarm probability $Q_F$ and detection probability $Q_D$ are then given by\footnote{These expressions are valid under the assumption that $\alpha<0.5$. Later in Section~\ref{sec6}, we will generalize our result for any arbitrary $\alpha$.}

    \newcommand{\nchoosek}[2]{\left(\begin{array}{c}#1\\#2\end{array}\right)}
    \vspace*{-0.1in}
    \begin{equation}
    \label{qf}
     Q_{F} = \sum_{i = K}^{N} \nchoosek{N}{i} (\pi_{1,0})^i (1-\pi_{1,0})^{N-i}
     \end{equation}
and
     \begin{equation}
     \label{qd}
     Q_{D} = \sum_{i = K}^{N} \nchoosek{N}{i} (\pi_{1,1})^i (1-\pi_{1,1})^{N-i},
    \end{equation}
  where $\pi_{j0}$ and $\pi_{j1}$ are the conditional probabilities of $u_i=j$ given $H_0$ and $H_1$, respectively.
 Specifically, $\pi_{1,0}$ and $\pi_{1,1}$ can be calculated as
   \begin{equation}
   \label{equ1}
\pi_{1,0}=\alpha(P_{1,0}(1-P_f)+(1-P_{0,1})P_f)+(1-\alpha)P_f
\end{equation}
and
\begin{equation}
\label{equ2}
\pi_{1,1}=\alpha(P_{1,0}(1-P_d)+(1-P_{0,1})P_d)+(1-\alpha)P_d,
\end{equation}
where $\alpha$ is the fraction of Byzantine nodes.

The local probability of error as seen by the FC is defined as
\begin{equation}
P_e = P_{0} \pi_{1,0} + P_{1} \left( 1-\pi_{1,1} \right)
\end{equation}
and the system wide probability of error at the FC is given by
\begin{equation}
P_E = P_{0} Q_F + P_{1} \left( 1-Q_D \right).
\end{equation}
Notice that, the system wide probability of error $P_E$ is a function of the parameter $K$, which is under the control of the FC, and the parameters $(\alpha,P_{j,0},P_{j,1})$ are under the control of the attacker.

The FC and the Byzantines may or may not have knowledge of their opponent's strategy. We will analyze the problem of detection with Byzantine data under several different scenarios in the following sections. First, we will determine the minimum fraction of Byzantines needed to blind the decision fusion scheme.

%YH: Where will you use this table? It is not clear here and later.
%BK: just to show that we will be analyzing these four cases

\section{Critical Power to Blind the fusion Center}
\label{sec3}
In this section, we determine the minimum fraction of Byzantine nodes needed to make the FC ``blind'' and denote it by $\alpha_{blind}$. We say that the FC is blind if an adversary can make the data that the FC receives from the sensors such that no information is conveyed. In other words, the optimal detector at the FC cannot perform better than simply making the decision based on priors. 
\begin{lemma}
In Bayesian distributed detection, the minimum fraction of Byzantines needed to make the FC blind is $\alpha_{blind}=0.5$. 
\end{lemma}
\begin{proof}
In the Bayesian framework, we say that the FC is ``blind'', if the received data $\mathbf{u}$ does not provide any information about the hypotheses to the FC. That is, the condition to make the FC blind can be stated as
\begin{equation}
\label{initial-blind}
P(H_i|\mathbf{u})=P(H_i)\; \text{for}\; i=0,1.
\end{equation}
It can be seen that \eqref{initial-blind} is equivalent to
\begin{eqnarray*}
&&
P(H_i|\mathbf{u})=P(H_i)\\
&\Leftrightarrow&
\dfrac{P(H_i)P(\mathbf{u}|H_i)}{P(\mathbf{u})}=P(H_i)\\
&\Leftrightarrow&
P(\mathbf{u}|H_i)=P(\mathbf{u}).
\end{eqnarray*}
Thus, the FC becomes blind if the probability of receiving a given vector $\mathbf{u}$ is independent of the hypothesis present. In such a scenario, the best that the FC can do is to make decisions solely based on the priors, resulting in the most degraded performance
at the FC. Now, using the conditional i.i.d. assumption, under which observations at the nodes are conditionally independent and identically
distributed given the hypothesis, condition \eqref{initial-blind} to make the FC blind becomes $\pi_{1,1}=\pi_{1,0}$. This is true only when
\begin{equation*}
\alpha[P_{1,0}(P_f-P_d)+(1-P_{0,1})(P_d-P_f)]+(1-\alpha)(P_d-P_f)=0.
\end{equation*}
Hence, the FC becomes blind if
\begin{equation}
\label{blind}
\alpha=\dfrac{1}{(P_{1,0}+P_{0,1})}.
\end{equation}
%Having identified the condition as given in \eqref{blind} under which the Byzantine attack makes the greatest impact on
%the performance of the network (blind the FC), we identify the strategy that the attacker should employ in order to achieve this
%condition as follows.
$\alpha$ in \eqref{blind} is minimized when $P_{1,0}$  and $P_{0,1}$ both take their largest values, i.e., $P_{1,0}=P_{0,1}=1$.
Hence, $\alpha_{blind}=0.5$.
\end{proof}

Next, we investigate how the Byzantines can launch an attack optimally considering that the parameter $(K)$ is under the control of the FC. By assuming error probability to be our performance metric, we analyze the non-asymptotic regime. Observe that the probability of error is dependent on the fusion rule. This gives us an additional degree of freedom to analyze the Byzantine attack under different scenarios where the FC and the
Byzantines may or may not have knowledge of their opponent's strategies.

\section{Optimal Attacking Strategies without the knowledge of Fusion Rule}
\label{sec5}
In practice, the Byzantine attacker may not have the knowledge about the fusion rule, i.e., the value of $K$, used by the FC. In such scenarios, we obtain the optimal attacking strategy for Byzantines by maximizing the local probability of error as seen by the FC, which is independent of the fusion rule $K$. We formally state the problem as

\begin{equation*}
\begin{aligned}
& \underset{P_{1,0},P_{0,1}}{\text{maximize}}
& & P_0 \pi_{1,0}+P_1(1-\pi_{1,1}) \\
& \text{subject to}
& &  0 \leq P_{1,0}\leq 1\\
& & & 0 \leq P_{0,1}\leq 1\\
\end{aligned}
\tag{P1}\label{opt-P2}
\end{equation*}
To solve the problem, we analyze the properties of the objective function, $P_e=P_0 \pi_{1,0}+P_1(1-\pi_{1,1})$, with respect to $(P_{1,0},P_{0,1})$. Notice that
\begin{equation}
\dfrac{dP_e}{P_{1,0}}=P_0 \alpha (1-P_f)-P_1 \alpha (1-P_d)\label{n1}
\end{equation}
and
\begin{equation}
\dfrac{dP_e}{P_{0,1}}=-P_0\alpha P_f+P_1\alpha P_d. \label{n2}
\end{equation} 
By utilizing monotonicity properties of the objective function with respect to $P_{1,0}$ and $P_{0,1}$ (\eqref{n1} and \eqref{n2}), we present the solution of the Problem~\ref{opt-P2} in Table~\ref{tableop}. 
Notice that, when $\frac{P_d}{P_f}<\frac{P_0}{P_1}<\frac{1-P_d}{1-P_f}$, both \eqref{n1} and \eqref{n2} are less than zero. $P_e$  then becomes a strictly decreasing function of $P_{1,0}$ as well as $P_{0,1}$. Hence, to maximize $P_e$, the attacker needs to choose $(P_{1,0},P_{0,1})=(0,0)$. However, the condition $\frac{P_d}{P_f}<\frac{P_0}{P_1}<\frac{1-P_d}{1-P_f}$ holds iff $P_d<P_f$ and, therefore, is not admissible.
Similar arguments lead to the rest of results given in Table~\ref{tableop}. Note that, if there is an equality in the conditions mentioned in Table~\ref{tableop}, then the solution will not be unique. For example, $\left(\dfrac{dP_e}{P_{0,1}}=0\right)\Leftrightarrow\left(\dfrac{P_0}{P_1}=\dfrac{1-P_d}{1-P_f}\right)$ implies that the
$P_e$ is constant as a function of $P_{0,1}$. In other words, the attacker will be indifferent in choosing the parameter $P_{0,1}$ because any value of $P_{0,1}$ will result in the same probability of error.

\begin{table}
\begin{center}
\caption{Soultion Of Maximizing Local Error $P_e$ Problem }
\label{tableop}
\renewcommand{\arraystretch}{1.5}
\resizebox{6cm}{!} {
  \begin{tabular}{ l | c | r }
    \hline
    $P_{1,0}$ & $P_{0,1}$ & Condition \\ \hline
    0 & 0 & $\frac{P_d}{P_f}<\frac{P_0}{P_1}<\frac{1-P_d}{1-P_f}$ \\ 
    0 & 1 & $\frac{P_d}{P_f}>\frac{P_0}{P_1}<\frac{1-P_d}{1-P_f}$ \\ 
    1 & 0 & $\frac{P_d}{P_f}<\frac{P_0}{P_1}>\frac{1-P_d}{1-P_f}$ \\
    1 & 1 & $\frac{P_d}{P_f}>\frac{P_0}{P_1}>\frac{1-P_d}{1-P_f}$ \\
   \hline
  \end{tabular}
}
\end{center}
\end{table}
Next, to gain insight into the solution, we present illustrative examples that corroborate our results.

\subsection{Illustrative Examples}

\begin{figure*}[t]
\centering
\subfigure[] {
\includegraphics[height=0.25\textheight, width=0.4\textwidth]{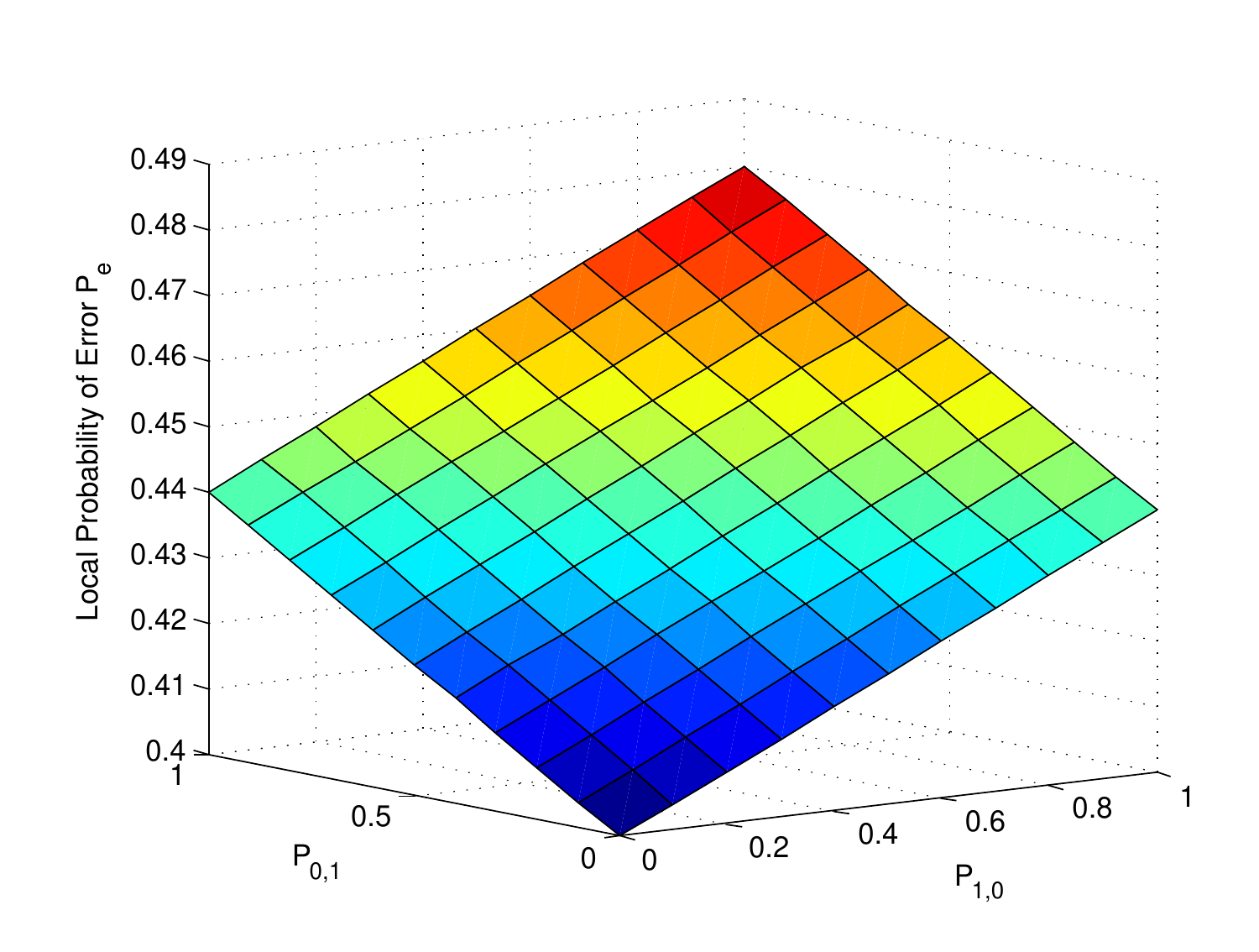}
\label{loca} }
\subfigure[]{
\includegraphics[height=0.25\textheight, width=0.4\textwidth]{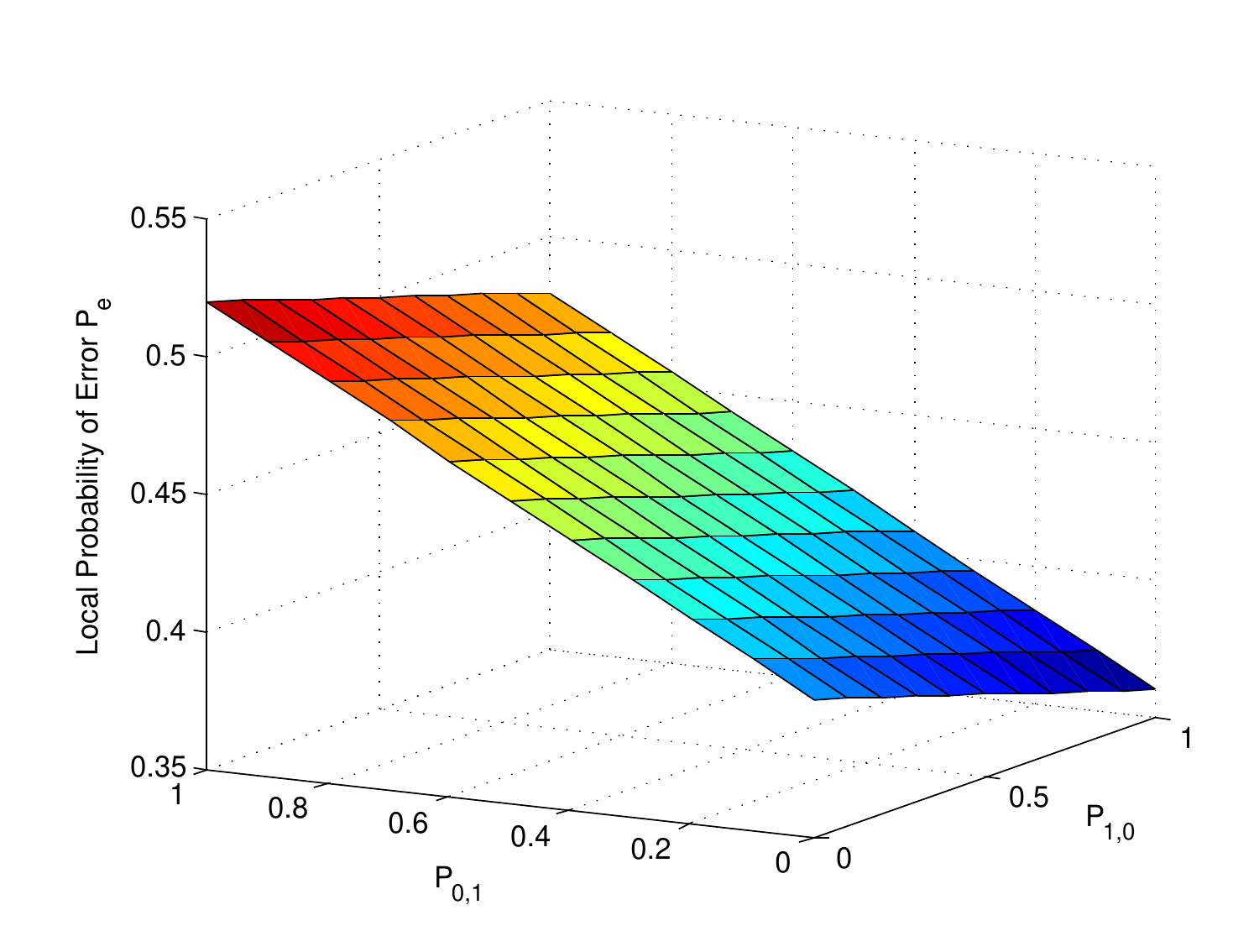}
\label{loca1}}
\caption{ \subref{loca} $P_e$ as a function of $(P_{1,0},P_{0,1})$ when $P_0=P_1=0.5$. \subref{loca1} $P_e$ as a function of $(P_{1,0},P_{0,1})$ when $P_0=0.1,P_1=0.9$.}
\label{local2}
\end{figure*}

In Figure~\ref{loca}, we plot the local probability of error $P_e$ as a function of $(P_{1,0},P_{0,1})$ when $(P_0=P_1=0.5)$.
We assume that the local probability of detection is $P_d=0.8$ and the local probability of false alarm is $P_f=0.1$ such that $\frac{P_d}{P_f}=8$, $\frac{1-P_d}{1-P_f}=.2222$, and $\frac{P_0}{P_1}=1$. Clearly, $\frac{P_d}{P_f}>\frac{P_0}{P_1}>\frac{1-P_d}{1-P_f}$ and it implies that the optimal attacking strategy is $(P_{1,0},P_{0,1})=(1,1)$, which can be verified from Figure~\ref{loca}.

In Figure~\ref{loca1}, we study the local probability of error $P_e$ as a function of the attacking strategy $(P_{1,0},P_{0,1})$ when $(P_0=0.1,P_1=0.9)$.
We assume that the local probability of detection is $P_d=0.8$ and the local probability of false alarm is $P_f=0.1$ such that $\frac{P_d}{P_f}=8$, $\frac{1-P_d}{1-P_f}=.2222$, and $\frac{P_0}{P_1}=.1111$. Clearly, $\frac{P_d}{P_f}>\frac{P_0}{P_1}<\frac{1-P_d}{1-P_f}$ implies that the optimal attacking strategy is $(P_{1,0},P_{0,1})=(0,1)$, which can be verified from the Figure~\ref{loca1}. These results corroborate our theoretical results presented in Table~\ref{tableop}.

In the next section, we investigate the scenario where Byzantines are aware of the fusion rule $K$ used at the FC and can use
this knowledge to provide false information in an optimal manner to blind the FC. However, the FC does not have knowledge of Byzantine's attacking strategies $(\alpha,P_{j,0},P_{j,1})$ and does not optimize against Byzantine's behavior. Since majority rule is a widely used fusion rule~\cite{Shi,Zhang,Kailkhura}, we assume that the FC uses the majority rule to make the global decision. 
 
\section{Optimal Byzantine Attacking Strategies with Knowledge of Majority Fusion Rule}
\label{majority}
In this section, we investigate optimal Byzantine attacking strategies in a distributed detection system, with the attacker having knowledge about the fusion rule used at the FC. However, we assume that the FC is not strategic in nature, and uses a majority rule, without trying to optimize against the Byzantine's behavior. We consider both the FC and the Byzantine to be strategic in Section~\ref{sec6}. The performance criterion at the FC is assumed to be the probability of error $P_E$.

For a fixed fusion rule $(K^*)$, which, as mentioned before, is assumed to be the majority rule $K^*=\lceil\frac{N+1}{2}\rceil$, $P_E$ varies with the parameters $(\alpha,P_{j,0},P_{j,1})$ which are under the control of the attacker. 
The Byzantine attack problem can be formally stated as follows:
\begin{equation}
\begin{aligned}
& \underset{P_{j,0},P_{j,1}}{\text{maximize}}
& & P_{E}(\alpha,P_{j,0},P_{j,1})\\
& \text{subject to}
& & 0 \leq P_{j,0}\leq 1\\
& & & 0 \leq P_{j,1}\leq 1. \\
\end{aligned}
\tag{P2}\label{problem}
\end{equation} 
For a fixed fraction of Byzantines $\alpha$, the attacker wants to maximize the probability of error $P_E$ by choosing its attacking strategy $(P_{j,0},P_{j,1})$ optimally. We assume that the attacker is aware of the fact that the FC is using the majority rule for making the global decision. Before presenting our main results for Problem~\ref{problem}, we make an assumption that will be used in the theorem.  

\begin{assumption}
\label{assumption}
We assume that $\alpha< \min\{(0.5-P_f),(1-(m/P_d))\}$,\footnote{Condition $\alpha< \min\{(0.5-P_f),(1-(m/P_d))\}$, where $m=\frac{N}{2N-2}>0.5$, suggests that as $N$ tends to infinity, $m=\dfrac{N}{2N-2}$ tends to $0.5$. When $P_d$ tends to 1 and $P_f$ tends to $0$, the above condition becomes $\alpha<0.5$.} where $m=\frac{N}{2N-2}$.
\end{assumption}
A consequence of this assumption is $\pi_{1,1}>m$, which can be shown as follows. By \eqref{equ2}, we have
 \begin{eqnarray}
 \pi_{1,1}&=&\alpha(P_{1,0}(1-P_d)+(1-P_{0,1})P_d)+(1-\alpha)P_d\nonumber\\
 &=&\alpha P_{1,0}(1-P_d) -\alpha P_d  P_{0,1}+P_d \nonumber\\
 &\geq &-\alpha P_d  P_{0,1}+P_d \geq P_d(1-\alpha)>m.\label{pi-m}
 \end{eqnarray} 
 Eq. \eqref{pi-m} is true because $\alpha< \min\{(0.5-P_f),(1-(m/P_d))\} \leq (1-(m/P_d))$. Another consequence of this assumption is $\pi_{1,0}<0.5$, which can be shown as follows. From \eqref{equ1}, we have
  \begin{eqnarray}
  \pi_{1,0}&=&\alpha(P_{1,0}(1-P_f)+(1-P_{0,1})P_f)+(1-\alpha)P_f\nonumber\\
  &=& \alpha P_{1,0}-\alpha P_f(P_{1,0}+P_{0,1})+P_f \nonumber\\
  &\leq & \alpha+P_f<0.5.\label{pi-m-2}
  \end{eqnarray} 
 Eq. \eqref{pi-m-2} is true because $\alpha< \min\{(0.5-P_f),(1-(m/P_d))\} \leq (0.5-P_f)$. 

Next, we analyze the properties of $P_E$ with respect to $(P_{1,0},P_{0,1})$ under our assumption that enable us to find the optimal attacking strategies. 
\begin{lemma}
            \label{Lemma-2}
            Assume that the FC employs the majority fusion rule $K^*$ and $\alpha< \min\{(0.5-P_f),(1-(m/P_d))\}$, where $m=\frac{N}{2N-2}$. 
          Then, for any fixed value of $P_{0,1}$, the error probability $P_E$ at the FC is a quasi-convex function of $P_{1,0}$.
          \end{lemma}
\begin{proof}
A function $f(P_{1,0})$ is quasi-convex if, for some $P_{1,0}^*$, $f(P_{1,0})$ is non-increasing
            for $P_{1,0} \leq P_{1,0}^*$ and $f(P_{1,0})$ is non-decreasing for $P_{1,0} \geq P_{1,0}^*$.
            In other words, the
            lemma is proved if $\dfrac{d P_E}{d P_{1,0}} \leq 0$ (or $\dfrac{d P_E}{d P_{1,0}} \geq 0$)
            for all $P_{1,0}$, or if for some $P_{1,0}^*$, $\dfrac{d P_E}{d P_{1,0}} \leq 0$
            when $P_{1,0} \leq P_{1,0}^*$ and
            $\frac{d P_E}{d P_{1,0}} \geq 0$ when $P_{1,0} \geq P_{1,0}^*$.
First, we calculate the partial derivative of $P_E$ with respect to
            $P_{1,0}$ for an arbitrary $K$ as follows: 
             \begin{equation}
                    \dfrac{d P_E}{d P_{1,0}} = P_0\dfrac{d Q_F}{d P_{1,0}} - P_1 \dfrac{d Q_D}{dP_{1,0}}.
\end{equation}
The detailed derivation of $\dfrac{d P_E}{d P_{1,0}}$ is given in Appendix~\ref{proof3} and we present a summary of the main results below.
\begin{equation}
\dfrac{d Q_F}{d P_{1,0}}= \alpha(1-P_f) N \nchoosek{N-1}{K-1} \left( \pi_{1,0} \right)^{K-1} \left( 1-\pi_{1,0} \right)^{N-K},         
\end{equation}
\begin{equation}
\dfrac{d Q_D}{d P_{1,0}}= \alpha(1-P_d) N \nchoosek{N-1}{K-1} \left( \pi_{1,1} \right)^{K-1} \left( 1-\pi_{1,1} \right)^{N-K},         
\end{equation}
and
\begin{equation*}
\dfrac{d P_E}{d P_{1,0}}=-P_{1}\alpha(1-P_d) N \nchoosek{N-1}{K-1} \left( \pi_{1,1} \right)^{K-1} \left( 1-\pi_{1,1} \right)^{N-K}
\end{equation*}
\begin{equation}
+P_{0}\alpha(1-P_f) N \nchoosek{N-1}{K-1} \left( \pi_{1,0} \right)^{K-1} \left( 1-\pi_{1,0} \right)^{N-K}. \label{derivative-P_E}
\end{equation}
$\dfrac{d P_E}{d P_{1,0}}$ given in \eqref{derivative-P_E} can be reformulated as follows:
            
         \begin{equation}
            \label{eqn6}
                \dfrac{d P_E}{d P_{1,0}} = g\left( P_{1,0}, K, \alpha \right)
                \left( e^{r\left( P_{1,0}, K, \alpha \right)} - 1 \right),
            \end{equation}
where
\begin{equation}
\label{gfunction}
g\left( P_{1,0}, K, \alpha \right)= N \nchoosek{N-1}{K-1}P_1 \alpha(1-P_d)( \pi_{1,1})^{K-1}( 1-\pi_{1,1})^{N-K}
\end{equation}                      
\text{and} 
\begin{eqnarray}
r\left( P_{1,0}, K, \alpha \right)&=& \ln{\left( \dfrac{P_0}{P_1}\dfrac{1-P_f}{1-P_d} \left( \dfrac{\pi_{1,0}}{\pi_{1,1}} \right)^{\left( K-1 \right)}\left( \dfrac{1-\pi_{1,0}}{1-\pi_{1,1}} \right)^{ \left( N-K \right) }\right)}\nonumber\\
&=& \ln\dfrac{P_0}{P_1}\dfrac{1-P_f}{1-P_d}+(K-1)\ln \dfrac{\pi_{1,0}}{\pi_{1,1}}+(N-K)\ln \dfrac{1-\pi_{1,0}}{1-\pi_{1,1}}. 
\end{eqnarray}
It can be seen that $g\left( P_{1,0}, K, \alpha \right) \geq 0$ so that
the sign of $\frac{d P_E}{d P_{1,0}}$ depends only on the value of $r\left( P_{1,0}, K, \alpha \right)$. 
To prove that $P_E$ is a quasi-convex function of $P_{1,0}$ when the majority rule $K^*$ is used at the FC, it is sufficient to show that $r\left( P_{1,0}, K^*, \alpha \right)$ is a non-decreasing function. Differentiating $r\left( P_{1,0}, K^*, \alpha \right)$ with respect to $P_{1,0}$, we get
            \begin{equation*}
               \frac{d r\left( P_{1,0}, K^*, \alpha \right) }{d P_{1,0}} = (K^*-1)\left(\dfrac{\alpha(1-P_f)}{\pi_{1,0}}-\dfrac{\alpha(1-P_d)}{\pi_{1,1}}\right)+(N-K^*)\left(\dfrac{\alpha(1-P_d)}{1-\pi_{1,1}}-\dfrac{\alpha(1-P_f)}{1-\pi_{1,0}}\right)
            \end{equation*}
            \begin{equation}
                = (K^*-1)\alpha\left(\dfrac{1-P_f}{\pi_{1,0}}-\dfrac{1-P_d}{\pi_{1,1}}\right)-(N-K^*)\alpha\left(\dfrac{1-P_f}{1-\pi_{1,0}}-\dfrac{1-P_d}{1-\pi_{1,1}}\right).
            \end{equation}
It can be shown that $ \dfrac{d r\left( P_{1,0}, K^*, \alpha \right) }{d P_{1,0}} > 0$ (see Appendix \ref{proof4}) and this completes the proof.
\end{proof} 
         
Quasi-convexity of $P_E$ over $P_{1,0}$ implies that the maximum of the function occurs on the corners, i.e., $P_{1,0}=\;0\;\text{or}\;1$ (may not be unique). Next, we analyze the properties of $P_E$ with respect to $P_{0,1}$.

\begin{lemma}
           \label{Lemma-3}
           Assume that the FC employs the majority fusion rule $K^*$ and $\alpha< \min\{(0.5-P_f),(1-(m/P_d))\}$, where $m=\frac{N}{2N-2}$. 
           Then, the probability of error $P_E$ at the FC is a quasi-convex function of $P_{0,1}$ for a fixed $P_{1,0}$.
\end{lemma}
 \begin{proof}          
 For a fixed $P_{1,0}$, we have                   
  \begin{equation}
 (\pi_{1,0})'=d\pi_{1,0}/dP_{0,1} =\alpha(-P_f).
 \end{equation}
By a similar argument as given in Appendix~\ref{proof3}, for an arbitrary $K$ we have
 \begin{equation*}
 \dfrac{d P_E}{d P_{0,1}}=P_{1}\alpha P_d N \nchoosek{N-1}{K-1} \left( \pi_{1,1} \right)^{K-1} \left( 1-\pi_{1,1} \right)^{N-K}
 \end{equation*}
 \begin{equation}
 -P_{0}\alpha P_f N \nchoosek{N-1}{K-1} \left( \pi_{1,0} \right)^{K-1} \left( 1-\pi_{1,0} \right)^{N-K}. \label{derivative-P_E-2}
 \end{equation}
$\dfrac{d P_E}{d P_{0,1}}$ given in \eqref{derivative-P_E-2} can be reformulated as follows:
             
          \begin{equation}
             \label{eqn7}
                 \dfrac{d P_E}{d P_{0,1}} = g\left( P_{0,1}, K, \alpha \right)
                 \left( e^{r\left( P_{0,1}, K, \alpha \right)} - 1 \right),
             \end{equation}
 where
 \begin{equation}
 g\left( P_{0,1}, K, \alpha \right)= N \nchoosek{N-1}{K-1}P_0 \alpha P_f( \pi_{1,0})^{K-1}( 1-\pi_{1,0})^{N-K}
 \end{equation}                      
 \text{and}
 \begin{eqnarray}
 r\left( P_{0,1}, K, \alpha \right)&=& \ln{\left( \dfrac{P_1}{P_0} \dfrac{P_d}{P_f} \left( \dfrac{\pi_{1,1}}{\pi_{1,0}} \right)^{\left( K-1 \right)}\left( \dfrac{1-\pi_{1,1}}{1-\pi_{1,0}} \right)^{ \left( N-K \right) }\right)}\nonumber\\
 &=& \ln\dfrac{P_1}{P_0}\dfrac{P_d}{P_f}+(K-1)\ln \dfrac{\pi_{1,1}}{\pi_{1,0}}+(N-K)\ln \dfrac{1-\pi_{1,1}}{1-\pi_{1,0}}. 
 \end{eqnarray}
It can be seen that $g\left( P_{0,1}, K, \alpha \right) \geq 0$ such that
 the sign of $\dfrac{d P_E}{d P_{0,1}}$ depends on the value of $r\left( P_{0,1}, K, \alpha \right)$.
 To prove that $P_E$ is a quasi-convex function of $P_{1,0}$ when the majority rule $K^*$ is used at the FC, it is sufficient to show that $r\left( P_{0,1}, K^*, \alpha \right)$ is a non-decreasing function.
 Differentiating $r\left( P_{0,1}, K^*, \alpha \right)$ with respect to $P_{0,1}$, we get
   \begin{equation}
                \dfrac{d r\left( P_{0,1}, K^*, \alpha \right) }{d P_{0,1}} = (K^*-1)\left(\dfrac{\alpha P_f}{\pi_{1,0}}-\dfrac{\alpha P_d}{\pi_{1,1}}\right)+(N-K^*)\left(\dfrac{\alpha P_d}{1-\pi_{1,1}}-\dfrac{\alpha P_f}{1-\pi_{1,0}}\right)
             \end{equation}
             \begin{equation}
                 = (N-K^*)\alpha\left(\dfrac{ P_d}{1-\pi_{1,1}}-\dfrac{ P_f}{1-\pi_{1,0}}\right)
 -(K^*-1)\alpha\left(\dfrac{P_d}{\pi_{1,1}}-\dfrac{ P_f}{\pi_{1,0}}\right).                
                 \end{equation}
 In the following, we show that
             \begin{equation}
             \label{dr1} 
             \dfrac{d r\left( P_{0,1}, K^*, \alpha \right) }{d P_{0,1}} > 0,
             \end{equation}
 i.e., $r\left( P_{0,1}, K^*, \alpha \right)$ is non-decreasing. It is sufficient to show that
\begin{equation}
\label{majocon}
(N-K^*)\left(\dfrac{ P_d}{1-\pi_{1,1}}-\dfrac{ P_f}{1-\pi_{1,0}}\right)
 >(K^*-1)\left(\dfrac{P_d}{\pi_{1,1}}-\dfrac{ P_f}{\pi_{1,0}}\right). 
\end{equation} 
First, we consider the case when there are an even number of nodes in the network and majority fusion rule is given by $K^*=\dfrac{N}{2}+1$. 
Since $0 \leq\pi_{1,0}<\pi_{1,1}\leq 1$ and $N\geq 2$, we have
 \begin{eqnarray}
 &&
\left(1-\dfrac{2}{N}\right)\dfrac{\pi_{1,1}\pi_{1,0}}{(1-\pi_{1,1})(1-\pi_{1,0})}>-1\nonumber\\
&\Leftrightarrow &
\left(1-\dfrac{2}{N}\right) \left[\dfrac{1}{1-\pi_{1,1}}-\dfrac{1}{1-\pi_{1,0}} \right]>\left[\dfrac{1}{\pi_{1,1}}-\dfrac{1}{\pi_{1,0}} \right]\nonumber\\
  &\Leftrightarrow  &
  \left[\left(1-\dfrac{2}{N}\right)\dfrac{1}{1-\pi_{1,1}}-\dfrac{1}{\pi_{1,1}} \right]>\left[\left(1-\dfrac{2}{N}\right)\dfrac{1}{1-\pi_{1,0}}-\dfrac{1}{\pi_{1,0}} \right].\label{pi-11}
 \end{eqnarray}
 Using the fact that $ \dfrac{P_d}{P_f}>1$, $\pi_{1,1}>\frac{N}{2N-2}$, and $K^*=\dfrac{N}{2}+1$, \eqref{pi-11} becomes
\begin{eqnarray}
&&
 \dfrac{P_d}{P_f}\left[\left(1-\dfrac{2}{N}\right)\dfrac{1}{1-\pi_{1,1}}-\dfrac{1}{\pi_{1,1}} \right]>\left[\left(1-\dfrac{2}{N}\right)\dfrac{1}{1-\pi_{1,0}}-\dfrac{1}{\pi_{1,0}} \right]\nonumber \\
 &\Leftrightarrow &
\left(1-\dfrac{2}{N}\right)\dfrac{ P_d}{1-\pi_{1,1}}-\dfrac{P_d}{\pi_{1,1}}> \left(1-\dfrac{2}{N}\right)\dfrac{P_f}{1-\pi_{1,0}}-\dfrac{P_f}{\pi_{1,0}}\nonumber \\
% &\Leftrightarrow &
%\left(\dfrac{N}{2}-1\right)\left(\dfrac{ P_d}{1-\pi_{1,1}}-\dfrac{ P_f}{1-\pi_{1,0}}\right)
% >\dfrac{N}{2}\left(\dfrac{P_d}{\pi_{1,1}}-\dfrac{ P_f}{\pi_{1,0}}\right)\label{evnmajo}\\
&\Leftrightarrow &
(N-K^*)\left(\dfrac{ P_d}{1-\pi_{1,1}}-\dfrac{ P_f}{1-\pi_{1,0}}\right)
 >(K^*-1)\left(\dfrac{P_d}{\pi_{1,1}}-\dfrac{ P_f}{\pi_{1,0}}\right). 
 \label{evnmajo}
 \end{eqnarray}
 
Next, we consider the case when there are odd number of nodes in the network and majority fusion rule is given by $K^*=\dfrac{N+1}{2}$. 
By using the fact that $\frac{\pi_{1,0}}{\pi_{1,1}}>\frac{P_f}{P_d}$, it can be seen that the right-hand side of \eqref{evnmajo} is nonnegative. Hence, from \eqref{evnmajo}, we have
\begin{eqnarray}
&&
\left(\dfrac{N}{2}-1\right)\left(\dfrac{ P_d}{1-\pi_{1,1}}-\dfrac{ P_f}{1-\pi_{1,0}}\right)
 >\dfrac{N}{2}\left(\dfrac{P_d}{\pi_{1,1}}-\dfrac{ P_f}{\pi_{1,0}}\right)\nonumber\\
&\Leftrightarrow&
 \left(\dfrac{N-1}{2}\right)\left(\dfrac{ P_d}{1-\pi_{1,1}}-\dfrac{ P_f}{1-\pi_{1,0}}\right)
    >  \left(\dfrac{N-1}{2}\right)\left(\dfrac{ P_d}{1-\pi_{1,1}}-\dfrac{ P_f}{1-\pi_{1,0}}\right)\nonumber\\
&\Leftrightarrow&
(N-K^*)\left(\dfrac{ P_d}{1-\pi_{1,1}}-\dfrac{ P_f}{1-\pi_{1,0}}\right)
 >(K^*-1)\left(\dfrac{P_d}{\pi_{1,1}}-\dfrac{ P_f}{\pi_{1,0}}\right). 
 \nonumber
\end{eqnarray}
This completes our proof.
% \begin{eqnarray*}
%      &&      (N-K^*)\left(\dfrac{ P_d}{1-\pi_{1,1}}-\dfrac{ P_f}{1-\pi_{1,0}}\right)
% >(K^*-1)\left(\dfrac{P_d}{\pi_{1,1}}-\dfrac{ P_f}{\pi_{1,0}}\right)   \\
% &\Leftrightarrow  &         
% P_d\left(\dfrac{N-K^*}{1-\pi_{1,1}}-\dfrac{K^*-1}{\pi_{1,1}}\right)>P_f\left(\dfrac{N-K^*}{1-\pi_{1,0}}-\dfrac{K^*-1}{\pi_{1,0}}\right)  \\                       
% &\Leftrightarrow  &    
% P_d\left(\dfrac{\pi_{1,1}(N-1)-(K^*-1)}{\pi_{1,1}(1-\pi_{1,1})}\right)>P_f\left(\dfrac{\pi_{1,0}(N-1)-(K^*-1)}{\pi_{1,0}(1-\pi_{1,0})}\right)\\
% &\Leftrightarrow  &    \dfrac{P_d}{P_f}\dfrac{\pi_{1,0}(1-\pi_{1,0})}{\pi_{1,1}(1-\pi_{1,1})}\left(\pi_{1,1}-\dfrac{K^*-1}{N-1}\right)
% >\left(\pi_{1,0}-\dfrac{K^*-1}{N-1}\right).
% \end{eqnarray*}
% 
% Observe that $\pi_{1,0}<\dfrac{K^*-1}{N-1}$(see \eqref{less}) implies right side term to be negative.\\
% Now if $\pi_{1,1}>\dfrac{K^*-1}{N-1}$, left hand side becomes positive and the inequality is true. \\
% Similarly $\pi_{1,1}>0.5\geq\dfrac{K^*-1}{N-1}$ can be shown if
% \begin{eqnarray}
% &&\alpha P_{1,0}-\alpha P_d(P_{1,0}+P_{0,1})+P_d
% =\alpha P_{1,0}(1-P_d) -\alpha P_d  P_{0,1}+P_d \nonumber\\
% &\geq &-\alpha P_d  P_{0,1}+P_d \geq P_d(1-\alpha)>0.5
% \end{eqnarray} 
% or in other words $\pi_{1,1}>0.5$ if $P_d(1-\alpha)>0.5$ or $\alpha <\left(1-\dfrac{0.5}{P_d}\right)$. It is true because $\alpha < \min\left\{(m-P_f),0.5\left(1-\dfrac{0.5}{P_d}\right)\right\}\leq \left(1-\dfrac{0.5}{P_d}\right)$.
 \end{proof}
 
 \begin{theorem}
             \label{Lemma-4}
   $(1,0)$, $(0,1)$, or $(1,1)$ are the optimal attacking strategies $(P_{1,0},P_{0,1})$ that maximize the probability of error $P_E$, when the majority fusion rule is employed at the FC and $\alpha< \min\{(0.5-P_f),(1-(m/P_d))\}$, where $m=\frac{N}{2N-2}$.
           \end{theorem} 
 \begin{proof}
 Lemma \ref{Lemma-2} and  Lemma \ref{Lemma-3} suggest that one of the corners is  the maximum of $P_E$   because of quasi-convexity. Note that $(0,0)$ cannot be the solution of the maximization problem since the attacker does not flip any results. Hence, we end up with three possibilities: $(1,0)$, $(0,1)$, or $(1,1)$. 
 \end{proof} 
 
 Next, to gain insights into Theorem~\ref{Lemma-4}, we present illustrative examples that corroborate our results.
 
 \subsection{Illustrative Examples}
 \begin{figure*}[t]
  \centering
  \subfigure[] {
  \includegraphics[%
    width=0.425\textwidth,clip=true]{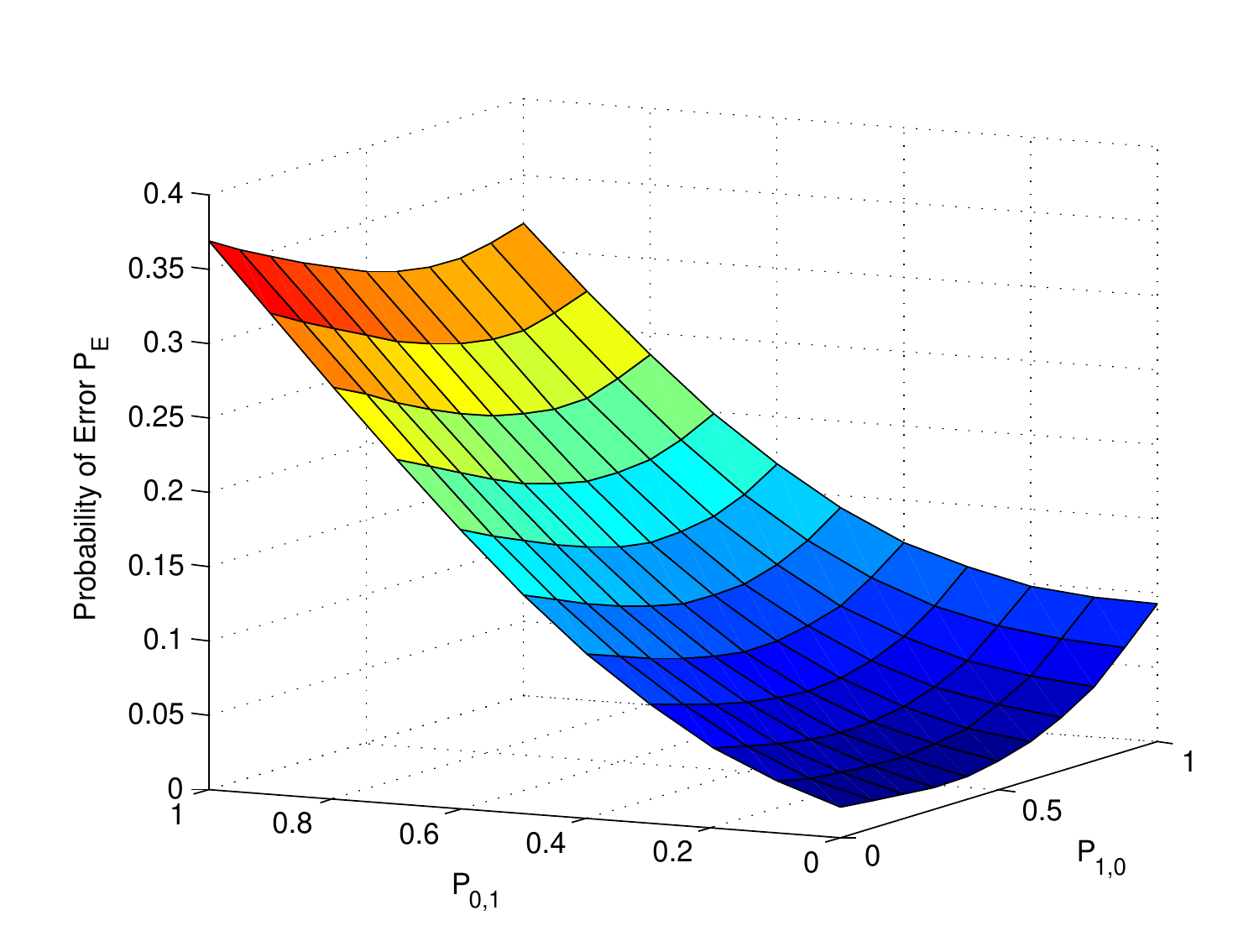}
  \label{even} }
  \subfigure[]{
  \includegraphics[%
    width=0.425\textwidth,clip=true]{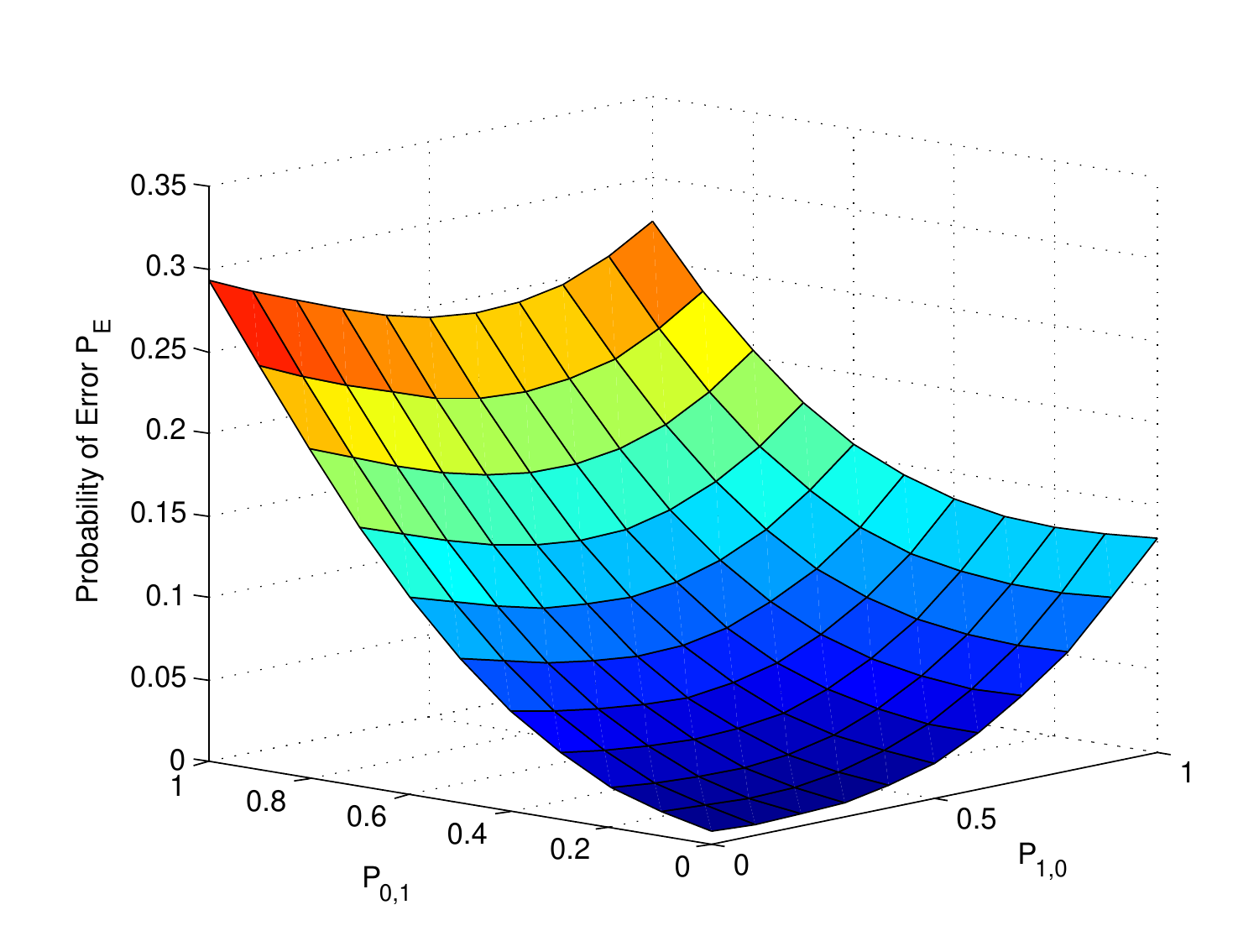}
  \label{odd}}
  \caption{ \subref{even} $P_E$ as a function of $(P_{1,0},P_{0,1})$ for $N=10$. \subref{odd} $P_E$ as a function of $(P_{1,0},P_{0,1})$ for $N=11$.}
  \label{quasi}
  \end{figure*}

In Figure~\ref{even}, we plot the probability of error $P_E$ as a function of the attacking strategy $(P_{1,0},P_{0,1})$ for even number of nodes, $N=10$, in the network.
We assume that the probability of detection is $P_d=0.8$, the probability of false alarm is $P_f=0.1$, prior probabilities are $(P_0=0.4,P_1=0.6)$, and $\alpha=0.37$. Since $\alpha< \min\{(0.5-P_f),(1-(m/P_d))\}$, where $m=\frac{N}{2N-2}$,
%$\alpha< \min\left\{\left(\dfrac{n-1}{2n-1}-P_f\right),\left(1-\dfrac{0.5}{P_d}\right)\right\}$, 
quasi-convexity can be observed in Figure~\ref{even}. 
Figure~\ref{odd} shows the probability of error $P_E$ as a function of attacking strategy $(P_{1,0},P_{0,1})$ for odd number of nodes, $N=11$, in the network. Similarly,  quasi-convexity can be observed in  Figure~\ref{odd}.

It is evident from Figures~\ref{even} and \ref{odd} that the optimal attacking strategy $(P_{1,0},P_{0,1})$ is either of the following three possibilities: $(1,0)$, $(0,1)$, or $(1,1)$. These results corroborate our theoretical results presented in Theorem~\ref{Lemma-4}.

Observe that the results obtained for this case are not the same as the results obtained for the asymptotic case (Please see Theorem~\ref{th1}). This is because the asymptotic performance measure (i.e., Chernoff information) is the exponential decay rate of the error probability of the ``optimal detector''. In other words, while optimizing over Chernoff information, one implicitly assumed that the optimal fusion rule is used at the FC.

Next, we investigate the case where the FC has the knowledge of attacker's strategies and uses the optimal fusion rule $K^*$ to make the global decision. Here, the attacker tries to maximize its worst case probability of error $\underset{K}{min} P_E$ by choosing $(P_{1,0},P_{0,1})$ optimally.
        
\section{Optimal Byzantine Attacking Strategies with Strategy-aware FC}
\label{sec6}
In this section, we analyze the scenario where the FC has the knowledge of attacker's strategies and uses the optimal fusion rule $K^*$ to make the global decision. The Byzantine attack problem can be formally stated as follows:
\begin{equation}
\begin{aligned}
& \underset{P_{j,0},P_{j,1}}{\text{maximize}}
& & P_{E}(K^{*},\alpha,P_{j,0},P_{j,1})\\
& \text{subject to}
& & 0 \leq P_{j,0}\leq 1\\
& & & 0 \leq P_{j,1}\leq 1, \\
\end{aligned}
\tag{P3}\label{opt-P4}
\end{equation} 
where $K^{*}$ is the optimal fusion rule. In other words, $K^{*}$ is the best response of the FC to the Byzantine attacking strategies. Next, we find the expression for the optimal fusion rule $K^*$ used at the FC.
\subsection{Optimal Fusion Rule}
First, we design the optimal fusion rule assuming that the local sensor threshold $\lambda$ and the Byzantine attacking strategy $(\alpha,P_{1,0},P_{0,1})$ are fixed and known to the FC.
\begin{lemma}
\label{kopt}
For a fixed local sensor threshold $\lambda$ and $\alpha<\dfrac{1}{P_{0,1}+P_{1,0}}$, the optimal fusion rule is given by
\begin{equation}
K^{*} \quad \mathop{\stackrel{H_1}{\gtrless}}_{H_0} \quad \dfrac{\ln\left[ (P_0/P_1)\left\{(1-\pi_{1,0})/(1-\pi_{1,1})\right\}^{N}\right]} {\ln \left[\{\pi_{1,1}(1-\pi_{1,0})\}/\{ \pi_{1,0} (1-\pi_{1,1})\}\right]}\label{K1*}.
\end{equation}
\end{lemma}
\begin{proof}
Consider the maximum a \textit{posteriori} probability (MAP) rule 
\begin{equation*}
\dfrac{P(\mathbf{u}|H_1)}{P(\mathbf{u}|H_0)} \quad \mathop{\stackrel{H_1}{\gtrless}}_{H_0} \quad  \dfrac{P_0}{P_1}.
\end{equation*}
Since the $u_i$’s are independent of each other, the MAP rule simplifies to
\begin{equation*}
\prod_{i=1}^{N}\dfrac{P({u_i}|H_1)}{P({u_i}|H_0)} \quad \mathop{\stackrel{H_1}{\gtrless}}_{H_0} \quad  \dfrac{P_0}{P_1}.
\end{equation*}
Let us assume that $K^*$ out of $N$ nodes send $u_i=1$. Now, the above equation can be written as 
\begin{equation*}
\dfrac{\pi_{1,1}^{K^*}(1-\pi_{1,1})^{N-K^*}}{\pi_{1,0}^{K^*}(1-\pi_{1,0})^{N-K^*}} \quad \mathop{\stackrel{H_1}{\gtrless}}_{H_0} \quad  \dfrac{P_0}{P_1}.
\end{equation*}
Taking logarithms on both sides of the above equation, we have
\begin{eqnarray}
&&
{K^*}\ln\pi_{1,1}+({N-K^*})\ln(1-\pi_{1,1})-K^*\ln\pi_{1,0}-(N-K^*)\ln(1-\pi_{1,0})\quad \mathop{\stackrel{H_1}{\gtrless}}_{H_0} \quad \ln \dfrac{P_0}{P_1}\nonumber\\
&\Leftrightarrow&
K^*[\ln(\pi_{1,1}/\pi_{1,0})+\ln((1-\pi_{1,0})/(1-\pi_{1,1}))]\quad \mathop{\stackrel{H_1}{\gtrless}}_{H_0} \quad \ln \dfrac{P_0}{P_1}+N\ln((1-\pi_{1,0})/(1-\pi_{1,1}))\nonumber\\
&\Leftrightarrow&
K^*\quad \mathop{\stackrel{H_1}{\gtrless}}_{H_0} \quad \dfrac{\ln \dfrac{P_0}{P_1}+N\ln((1-\pi_{1,0})/(1-\pi_{1,1}))}{[\ln(\pi_{1,1}/\pi_{1,0})+\ln((1-\pi_{1,0})/(1-\pi_{1,1}))]}\label{lemma-5-1}\\
&\Leftrightarrow&
K^{*} \quad \mathop{\stackrel{H_1}{\gtrless}}_{H_0} \quad \dfrac{\ln\left[ (P_0/P_1)\left\{(1-\pi_{1,0})/(1-\pi_{1,1})\right\}^{N}\right]} {\ln \left[\{\pi_{1,1}(1-\pi_{1,0})\}/\{ \pi_{1,0} (1-\pi_{1,1})\}\right]},\nonumber
\end{eqnarray}
where \eqref{lemma-5-1} follows from the fact that, for $\pi_{1,1}>\pi_{1,0}$ or equivalently, $\alpha<\dfrac{1}{P_{0,1}+P_{1,0}}$, $[\ln(\pi_{1,1}/\pi_{1,0})+\ln((1-\pi_{1,0})/(1-\pi_{1,1}))]>0$.
\end{proof}
The probability of false alarm $Q_F$ and the probability of detection $Q_D$ for this case are as given in \eqref{qf} and \eqref{qd} with $K=\left \lceil K^* \right \rceil$.
Next, we present our results for the case when the fraction of Byzantines $\alpha>\dfrac{1}{P_{0,1}+P_{1,0}}$.

\begin{lemma}
For a fixed local sensor threshold $\lambda$ and $\alpha>\dfrac{1}{P_{0,1}+P_{1,0}}$, the optimal fusion rule is given by
\begin{equation}
K^* \quad \mathop{\stackrel{H_0}{\gtrless}}_{H_1} \quad \dfrac{\ln\left[(P_1/P_0)\left\{(1-\pi_{1,1})/(1-\pi_{1,0})\right\}^{N}\right]}{[\ln(\pi_{1,0}/\pi_{1,1})+\ln((1-\pi_{1,1})/(1-\pi_{1,0}))]}\label{K2*}.
\end{equation}
\end{lemma}
\begin{proof}
This can be proved similarly as Lemma~\ref{kopt} and using the fact that, for $\pi_{1,1}<\pi_{1,0}$ or equivalently, $\alpha>\dfrac{1}{P_{0,1}+P_{1,0}}$, $[\ln(\pi_{1,0}/\pi_{1,1})+\ln((1-\pi_{1,1})/(1-\pi_{1,0}))]>0$.
\end{proof}
The probability of false alarm $Q_F$ and the probability of detection $Q_D$ for this case can be calculated to be
 \begin{equation}
    \label{qf1}
     Q_{F} = \sum_{i = 0}^{\left \lfloor K^* \right \rfloor} \nchoosek{N}{i} (\pi_{1,0})^i (1-\pi_{1,0})^{N-i}
     \end{equation}
and
     \begin{equation}
     \label{qd1}
     Q_{D} = \sum_{i = 0}^{\left \lfloor K^* \right \rfloor} \nchoosek{N}{i} (\pi_{1,1})^i (1-\pi_{1,1})^{N-i}.
    \end{equation}
    
    Next, we analyze the property of $P_E$ with respect to Byzantine attacking strategy $(P_{1,0},P_{0,1})$ that enables us to find the optimal attacking strategies.

\begin{lemma}
            \label{Lemma-7}
            For a fixed local sensor threshold $\lambda$, assume that the FC employs the optimal fusion rule $\left \lceil K^* \right \rceil$, \footnote{Notice that, $K^*$ might not be an integer.} as given in \eqref{K1*}. Then, for $\alpha\leq 0.5$, the error probability $P_E$ at the FC is a monotonically increasing function of $P_{1,0}$ while $P_{0,1}$ remains fixed. Conversely, the error probability $P_E$ at the FC is a monotonically increasing function of $P_{0,1}$ while $P_{1,0}$ remains fixed.
          \end{lemma}
%YH2: it is okay to go with monotonically in this lemma. Hence, I keep it.
%BK: I agree prof Han.
\begin{proof}
Observe that, for a fixed $\lambda$, $P_E(\left \lceil K^* \right \rceil)$ is a continuous but not a differentiable function. However, the function is non differentiable only at a finite number (or infinitely countable number) of points because of the nature of $\left \lceil K^* \right \rceil$. 
Now observe that, for a fixed fusion rule $K$, $P_E(K)$ is differentiable. Utilizing this fact, to show that the lemma is true, we first find the condition that a fusion rule $K$ should satisfy so that $P_E$ is a monotonically increasing function of $P_{1,0}$ while keeping $P_{0,1}$ fixed (and vice versa) and later show that $\left \lceil K^* \right \rceil$ satisfies this condition. From \eqref{eqn6}, finding those $K$ that satisfy $\dfrac{d P_E}{d P_{1,0}}> 0$\footnote{Observe that, for $\alpha< 0.5$, the function $g\left( P_{1,0}, K^*, \alpha \right)=0$ (as given in \eqref{gfunction}) only under extreme conditions (i.e., $P_1=0$ or $P_d=0$ or $P_d=1$). Ignoring these extreme conditions, we have $g\left( P_{1,0}, K^*, \alpha \right)>0$.} is equivalent to finding those value of $K$ that make
\begin{eqnarray}
&&
r\left( P_{1,0}, K, \alpha \right)>0\nonumber\\ 
&\Leftrightarrow  &
\ln\dfrac{P_0}{P_1}\dfrac{1-P_f}{1-P_d}+(K-1)\ln \dfrac{\pi_{1,0}}{\pi_{1,1}}+(N-K)\ln \dfrac{1-\pi_{1,0}}{1-\pi_{1,1}}>0\nonumber \\
&\Leftrightarrow  &
K < \dfrac{\ln \dfrac{P_0}{P_1}+N\ln \dfrac{(1-\pi_{1,0})}{(1-\pi_{1,1})}+\ln\dfrac{1-P_f}{1-P_d}-\ln\dfrac{\pi_{1,0}}{\pi_{1,1}}} {\ln \left[\{\pi_{1,1}(1-\pi_{1,0})\}/\{ \pi_{1,0} (1-\pi_{1,1})\}\right]}.\label{k1}
\end{eqnarray}
Similarly, we can find the condition that a fusion rule $K$ should satisfy so that $P_E$ is a monotonically increasing function of $P_{0,1}$ while keeping $P_{1,0}$ fixed. From \eqref{eqn7}, finding those $K$ that satisfy $\dfrac{d P_E}{d P_{0,1}}> 0$ is equivalent to finding those $K$ that make
\begin{eqnarray}
&&
r\left( P_{0,1}, K, \alpha \right)>0\nonumber\\ 
&\Leftrightarrow  &
\ln\dfrac{P_1}{P_0}\dfrac{P_d}{P_f}+(K-1)\ln \dfrac{\pi_{1,1}}{\pi_{1,0}}+(N-K)\ln \dfrac{1-\pi_{1,1}}{1-\pi_{1,0}}>0 \nonumber\\
&\Leftrightarrow  &
K > \dfrac{\ln \dfrac{P_0}{P_1}+N\ln \dfrac{(1-\pi_{1,0})}{(1-\pi_{1,1})}+\ln\dfrac{P_f}{P_d}-\ln\dfrac{\pi_{1,0}}{\pi_{1,1}}} {\ln \left[\{\pi_{1,1}(1-\pi_{1,0})\}/\{ \pi_{1,0} (1-\pi_{1,1})\}\right]}.\label{k2}
\end{eqnarray}
From \eqref{k1} and \eqref{k2}, we have

{\small \begin{equation}
A=\dfrac{\ln \dfrac{P_0}{P_1}+N\ln \dfrac{(1-\pi_{1,0})}{(1-\pi_{1,1})}+\ln\dfrac{1-P_f}{1-P_d}-\ln\dfrac{\pi_{1,0}}{\pi_{1,1}}} {\ln \left[\{\pi_{1,1}(1-\pi_{1,0})\}/\{ \pi_{1,0} (1-\pi_{1,1})\}\right]}>K > \dfrac{\ln \dfrac{P_0}{P_1}+N\ln \dfrac{(1-\pi_{1,0})}{(1-\pi_{1,1})}+\ln\dfrac{P_f}{P_d}-\ln\dfrac{\pi_{1,0}}{\pi_{1,1}}} {\ln \left[\{\pi_{1,1}(1-\pi_{1,0})\}/\{ \pi_{1,0} (1-\pi_{1,1})\}\right]}=B.\label{k}
\end{equation}}
Next, we show that the optimal fusion rule $\left \lceil K^* \right \rceil$  given in \eqref{K1*} is within the region $(A,B)$. First we prove that $\left \lceil K^* \right \rceil>B$ by showing $K^*>B$. Comparing $K^*$ given in~\eqref{K1*} with $B$, $K^*>B$ iff 
\begin{equation}
\label{ratio_P_f_P_d}
0>\ln\dfrac{P_f}{P_d}-\ln\dfrac{\pi_{1,0}}{\pi_{1,1}}.
\end{equation}
Since $P_d>P_f$, to prove \eqref{ratio_P_f_P_d} we start from the inequality 
\begin{eqnarray*}
&&
\dfrac{(1-P_d)}{P_d}< \dfrac{(1-P_f)}{P_f}\\
&\Leftrightarrow  &
\dfrac{\alpha P_{1,0}(1-P_d)+P_d(1-P_{0,1}\alpha)}{P_d}<\dfrac{\alpha P_{1,0}(1-P_f)+P_f(1-P_{0,1}\alpha)}{P_f}\\
&\Leftrightarrow  &
\dfrac{\pi_{1,1}}{P_d}<\dfrac{\pi_{1,0}}{P_f}\\
&\Leftrightarrow  &
0>\ln\dfrac{P_f}{P_d}-\ln\dfrac{\pi_{1,0}}{\pi_{1,1}}.
\end{eqnarray*}

Now, we show that  $A>\left \lceil K^* \right \rceil$. Observe that,
\begin{eqnarray*}
&&
A>\left \lceil K^* \right \rceil\\
&\Leftrightarrow  &
\dfrac{\ln\dfrac{1-P_f}{1-P_d}-\ln\dfrac{\pi_{1,0}}{\pi_{1,1}}} {\ln \left[\{\pi_{1,1}(1-\pi_{1,0})\}/\{ \pi_{1,0} (1-\pi_{1,1})\}\right]}>\left \lceil K^* \right \rceil-K^*.
\end{eqnarray*}

Hence, it is sufficient to show that
\begin{equation*}
\dfrac{\ln\dfrac{1-P_f}{1-P_d}-\ln\dfrac{\pi_{1,0}}{\pi_{1,1}}} {\ln \left[\{\pi_{1,1}(1-\pi_{1,0})\}/\{ \pi_{1,0} (1-\pi_{1,1})\}\right]}>1>\left \lceil K^* \right \rceil-K^*.
\end{equation*}
$1>\left \lceil K^* \right \rceil-K^*$ is true from the property of the ceiling function. By \eqref{a1}, we have
\begin{eqnarray*}
&&
\dfrac{1-P_f}{1-P_d}>\dfrac{1-\pi_{1,0}}{1-\pi_{1,1}}\\
&\Leftrightarrow  &
\ln\dfrac{1-P_f}{1-P_d}>\ln\dfrac{1-\pi_{1,0}}{1-\pi_{1,1}}\\
&\Leftrightarrow  &
\ln\dfrac{1-P_f}{1-P_d}-\ln\dfrac{\pi_{1,0}}{\pi_{1,1}}>\ln \left[\{\pi_{1,1}(1-\pi_{1,0})\}/\{ \pi_{1,0} (1-\pi_{1,1})\}\right]\\
&\Leftrightarrow  &
\dfrac{\ln\dfrac{1-P_f}{1-P_d}-\ln\dfrac{\pi_{1,0}}{\pi_{1,1}}} {\ln \left[\{\pi_{1,1}(1-\pi_{1,0})\}/\{ \pi_{1,0} (1-\pi_{1,1})\}\right]}>1
\end{eqnarray*}
which completes the proof.
\end{proof}
Based on  Lemma~\ref{Lemma-7}, we present the optimal attacking strategies for the case when the FC has the knowledge regarding the strategies used by the Byzantines.
\begin{theorem}
\label{th3}
The optimal attacking strategies, $(P_{1,0}^*, P_{0,1}^*)$, which maximize the probability of error, $P_E(\left \lceil K^* \right \rceil)$, are given by
\[ (P_{1,0}^*, P_{0,1}^*)  \left\{ \begin{array}{rll}
				(p_{1,0}, p_{0,1})  & \mbox{if}\ \alpha>0.5 \\
			   (1,1)  & \mbox{if}\ \alpha\leq 0.5
				\end{array}\right.
\] 
where $(p_{1,0}, p_{0,1})$ satisfies $\alpha(p_{1,0}+p_{0,1})=1$.
\end{theorem}
\begin{proof}
Note that, the maximum probability of error occurs when the posterior probabilities are equal to the prior probabilities of the hypotheses. That is,
\begin{equation}
\label{initial-blind1}
P(H_i|\mathbf{u})=P(H_i)\; \text{for}\; i=0,1.
\end{equation}
%It can be seen that \eqref{initial-blind1} is equivalent to
%\begin{eqnarray}
%&&
%P(H_i|\mathbf{u})=P(H_i)\nonumber\\
%&\Leftrightarrow&
%\dfrac{P(H_i)P(\mathbf{u}|H_i)}{P(\mathbf{u})}=P(H_i)\nonumber\\
%&\Leftrightarrow&
%P(\mathbf{u}|H_i)=P(\mathbf{u}).\label{initial-blind2}
%\end{eqnarray}
%Now, using the conditional i.i.d. assumption, \eqref{initial-blind2} becomes $\pi_{1,1}=\pi_{1,0}$, i.e.,   likelihoods under the two hypotheses are equal.
%YH2: I am not sured about the above argument. can you give a more rigirous proof?
%BK: the proof is same as \alpha_{blind} proof so I copied it
Now using the result from \eqref{blind}, the condition can be simplified to
\begin{equation}
\label{alpha_one}
\alpha (P_{1,0}+P_{0,1})=1.
\end{equation}
Eq. \eqref{alpha_one} suggests that when $\alpha\geq 0.5$, the attacker can find flipping probabilities that make $P_E=\min\{P_0,P_1\}$. When $\alpha=0.5$, $P_{1,0}=P_{0,1}=1$ is the optimal attacking strategy and when $\alpha>0.5$, any pair which satisfies $P_{1,0}+P_{0,1}=\dfrac{1}{\alpha}$ is optimal. However, when $\alpha<0.5$, \eqref{alpha_one} cannot be satisfied. In this case,  by Lemma~\ref{Lemma-7}, for $\alpha<0.5$, $(1,1)$ is an optimal attacking strategy, $(P_{1,0}, P_{0,1})$, which maximizes probability of error, $P_E(\left \lceil K^* \right \rceil)$.
\end{proof}
Next, to gain insight into Theorem~\ref{th3}, we present illustrative examples that corroborate our results.
\subsection{Illustrative Examples}
\begin{figure*}[t]
\centering
\subfigure[] {
\includegraphics[height=0.25\textheight, width=0.4\textwidth]{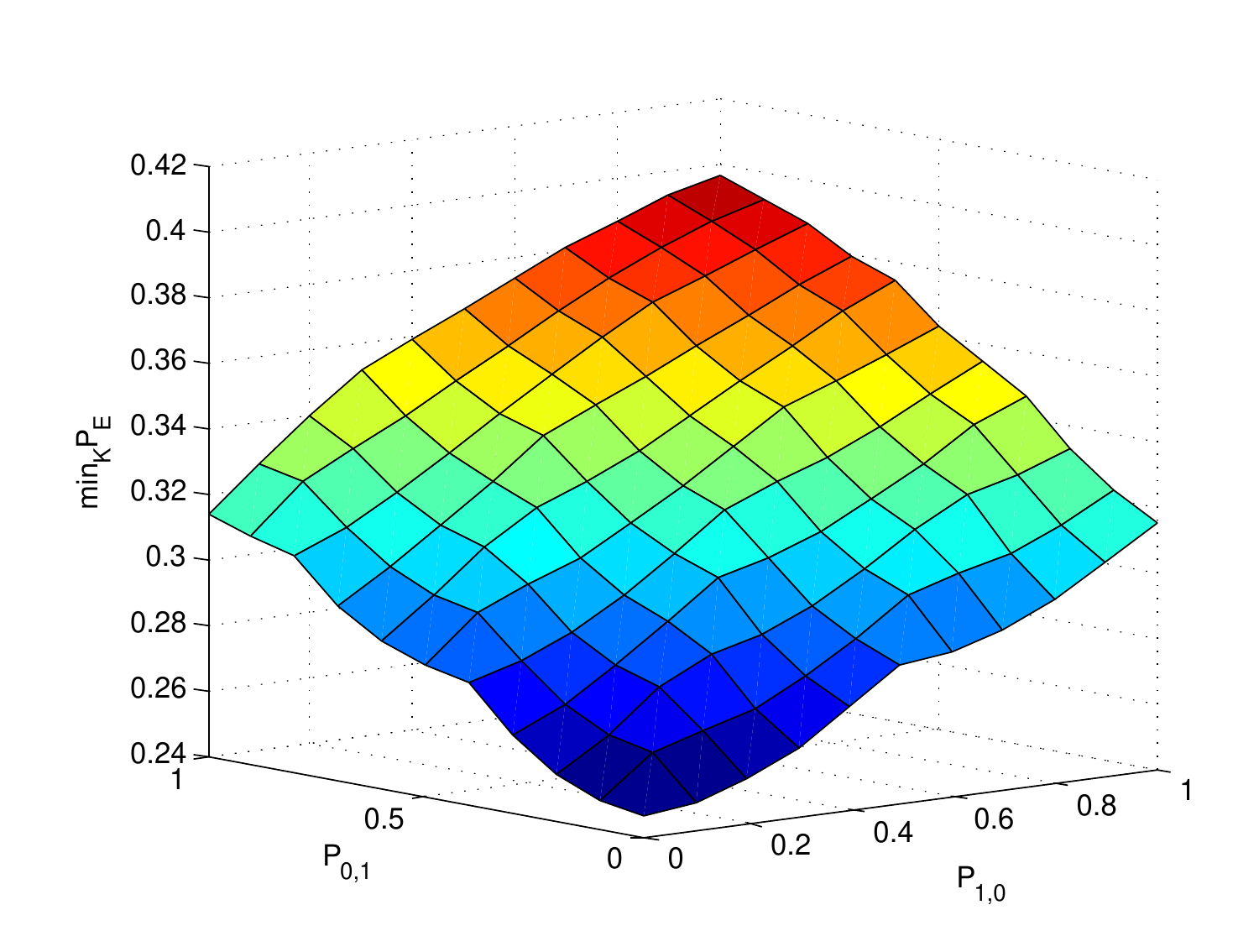}
\label{min} }
\subfigure[]{
\includegraphics[height=0.25\textheight, width=0.4\textwidth]{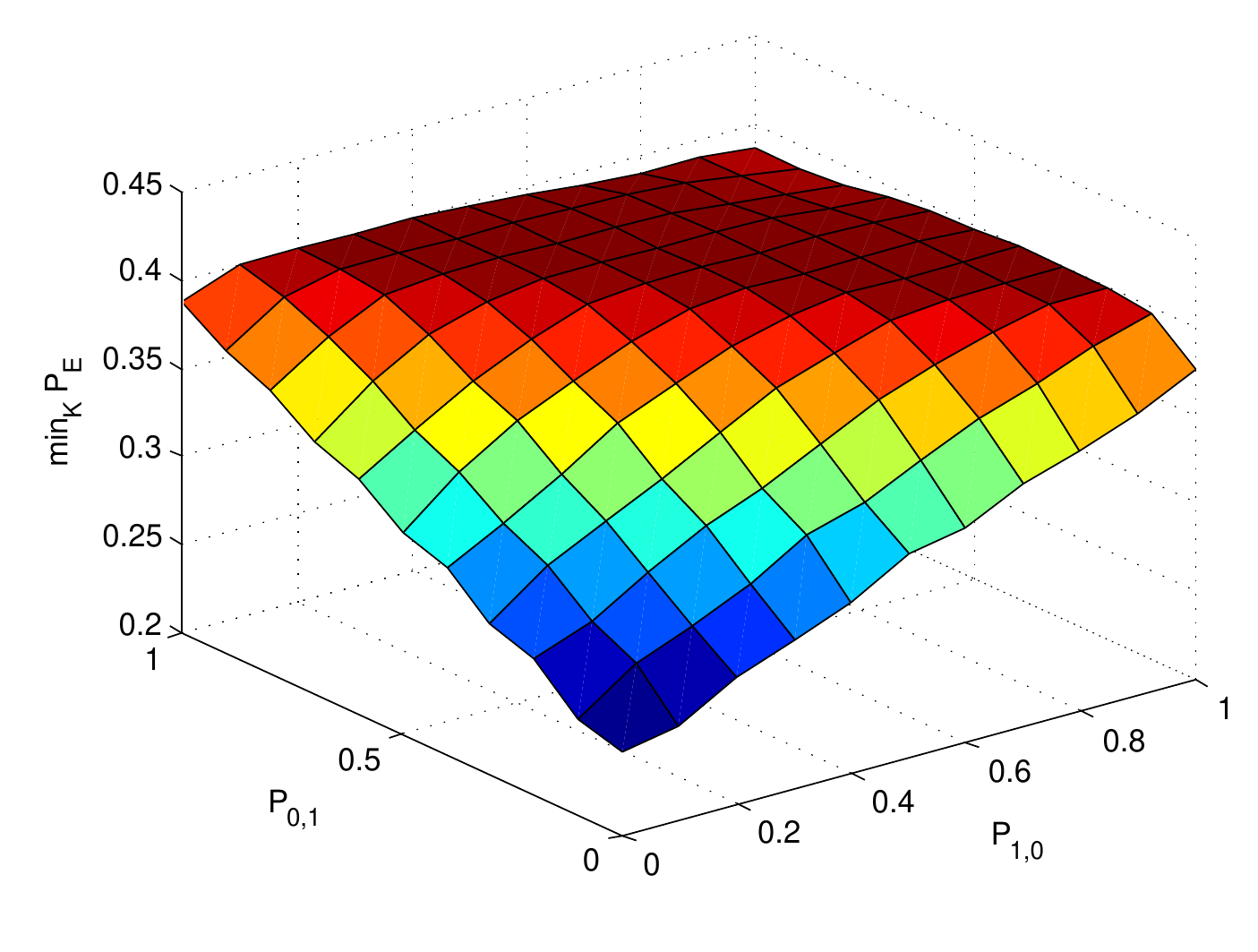}
\label{min1}}
\caption{Minimum probability of error ($\min_K P_E$) analysis. \subref{min} $\min_K P_E$ as a function of $(P_{1,0},P_{0,1})$ for $\alpha=0.4$. \subref{min1} $\min_K P_E$ as a function of $(P_{1,0},P_{0,1})$ for $\alpha=0.8$.}
\label{minpe}
\end{figure*}
In Figure~\ref{minpe}, we plot the minimum probability of error as a function of attacker's strategy $(P_{1,0},P_{0,1})$, where $P_E$ is minimized over all possible fusion rules $K$. We consider a $N=11$ node network, with the nodes' detection and false alarm probabilities being $0.6$ and $0.4$, respectively. Prior probabilities are assumed to be $P_0 = 0.4 $ and $P_1 = 0.6$.
%First, in Figure~\ref{min} we look into the scenario where $\alpha=0.4$ (i.e., $\alpha<0.5$). 
%Here, we have assumed that the probability of detection is $P_d=0.6$, the probability of false alarm is $P_f=0.4$, prior probabilities are $(P_0=0.4,P_1=0.6)$ and total $N=11$ nodes are in the network. 
Observe that, the optimal fusion rule as given in \eqref{K1*} changes with attacker's strategy $(P_{1,0},P_{0,1})$. Thus, the minimum probability of error $\min_K P_E$ is a non-differentiable function. It is evident from Figure~\ref{min} that $(P_{1,0},P_{0,1})=(1,1)$ maximizes the probability of error, $P_E(\left \lceil K^* \right \rceil)$. This corroborates our theoretical results presented in Theorem~\ref{th3},  that for $\alpha<0.5$, the optimal attacking strategy, $(P_{1,0}, P_{0,1})$, that maximizes the probability of error, $P_E(\left \lceil K^* \right \rceil)$, is $(1,1)$.

In Figure~\ref{min1} we consider the scenario where $\alpha=0.8$ (i.e., $\alpha>0.5$). It can be seen that the attacking strategy $(P_{1,0},P_{0,1})$, that maximizes $\min_K P_E$ is not unique in this case. It can be verified that any attacking strategy which satisfies $P_{1,0}+P_{0,1}=\frac{1}{0.8}$ will make $\min_K P_E=\min\{P_0,P_1\}=0.4$. This corroborates our theoretical results presented in Theorem~\ref{th3}.\\
Observe that the results obtained for this case are consistent with the results obtained for the asymptotic case. This is because the optimal fusion rule is used at the FC and the asymptotic performance measure (i.e., Chernoff information) is the exponential decay rate of error probability of the ``optimal detector'', and thus, implicitly assumes that the optimal fusion rule is used at the FC.

When the attacker does not have the knowledge of the fusion rule $K$ used at the FC,  from an attacker's perspective, maximizing its local probability of error $P_e$ is the optimal attacking strategy. The optimal attacking strategy in this case is either of the three possibilities: $(P_{1,0},P_{0,1})=(0,1)\;\text{or}\;(1,0)\;\text{or}\;(1,1)$ (see Table~\ref{tableop}). However, the FC has knowledge of the attacking strategy $(\alpha,P_{1,0},P_{0,1})$ and thus,  uses the optimal fusion rule as given in~\eqref{K1*} and ~\eqref{K2*}.

\section{Conclusion and Future Work}
\label{sec7}
We considered the problem of distributed Bayesian detection with Byzantine data, and characterized the power of attack analytically. 
For distributed detection for a binary hypothesis testing problem, the expression for the minimum attacking power above which the ability to detect is completely destroyed was obtained. We showed that when there are more than $50\%$ of Byzantines in the network,
the data fusion scheme becomes blind and no detector can achieve any performance gain over the one based just on priors. The optimal attacking strategies for Byzantines that degrade the performance at the FC were obtained. It was shown that the results obtained for the non-asymptotic case are consistent with the results obtained for the asymptotic case only when the FC has the knowledge of the attacker's
strategies, and thus, uses the optimal fusion rule. However, results obtained for the non-asymptotic
case, when the FC does not have knowledge of attacker's strategies, are not the same as the results obtained for the asymptotic case. There are still many interesting questions that
remain to be explored in the future work such as an analysis of
the scenario where Byzantines can also control sensor thresholds used for making local decisions. Other questions
such as the case where Byzantines collude in several groups (collaborate) to degrade the detection performance can also be investigated.

\section*{Acknowledgment}
This work was supported in part by ARO under Grant W911NF-14-1-0339, AFOSR under Grant FA9550-10-1-0458 and National Science Council of Taiwan, under grants NSC 99-2221-E-011-158 -MY3, NSC 101-2221-E-011-069 -MY3. Han's work was completed during his visit to Syracuse University from 2012 to 2013.

\appendices

\section{Proof of $\dfrac{d r\left( P_{1,0}, K^*, \alpha \right) }{d P_{1,0}} > 0$}
\label{proof4}
Differentiating both sides of $r\left( P_{1,0}, K^*, \alpha \right)$ with respect to $P_{1,0}$, we get
\begin{equation*}
       \frac{d r\left( P_{1,0}, K^*, \alpha \right) }{d P_{1,0}} =  (K^*-1)\alpha\left(\dfrac{1-P_f}{\pi_{1,0}}-\dfrac{1-P_d}{\pi_{1,1}}\right)-(N-K^*)\alpha\left(\dfrac{1-P_f}{1-\pi_{1,0}}-\dfrac{1-P_d}{1-\pi_{1,1}}\right).
\end{equation*} 
In the following we show that
            \begin{equation}
            \label{dr} 
            \dfrac{d r\left( P_{1,0}, K^*, \alpha \right) }{d P_{1,0}} > 0
            \end{equation}
i.e., $r\left( P_{1,0}, K^*, \alpha \right)$ is non-decreasing.  
Observe that in the above equation,
\begin{equation}
\label{equality}
\dfrac{(1-P_f)}{\pi_{1,0}}>\dfrac{(1-P_d)}{\pi_{1,1}}.
\end{equation}
To show that the above condition is true, we start from the inequality
\begin{eqnarray}
&&
P_d>P_f\\
&\Leftrightarrow&
\dfrac{P_d}{1-P_d}>\dfrac{P_f}{1-P_f}\\
&\Leftrightarrow&
\alpha P_{1,0}+(1-P_{0,1}\alpha)\dfrac{P_d}{1-P_d}>\alpha P_{1,0}+(1-P_{0,1}\alpha) \dfrac{P_f}{1-P_f}\\
&\Leftrightarrow&
\dfrac{\alpha P_{1,0}(1-P_d)+P_d(1-P_{0,1}\alpha)}{(1-P_d)}>\dfrac{\alpha P_{1,0}(1-P_f)+P_f(1-P_{0,1}\alpha)}{(1-P_f)}\\
&\Leftrightarrow&
\dfrac{\pi_{1,1}}{(1-P_d)}>\dfrac{\pi_{1,0}}{(1-P_f)}\\\label{a1}
&\Leftrightarrow&
\dfrac{(1-P_f)}{\pi_{1,0}}>\dfrac{(1-P_d)}{\pi_{1,1}}
\end{eqnarray}
Similarly, it can be shown that 
\begin{eqnarray}
&&\dfrac{1-\pi_{1,1}}{1-P_d}> \dfrac{1-\pi_{1,0}}{1-P_f}\label{greater1}
\end{eqnarray}
Now from \eqref{equality} and \eqref{greater1}, to show that  $\dfrac{d r\left( P_{1,0}, K^*, \alpha \right) }{d P_{1,0}}>0$ is equivalent to show that 
\begin{equation}
\label{majcon1}
 (K^*-1)\left(\dfrac{1-P_f}{\pi_{1,0}}-\dfrac{1-P_d}{\pi_{1,1}}\right)>(N-K^*)\left(\dfrac{1-P_f}{1-\pi_{1,0}}-\dfrac{1-P_d}{1-\pi_{1,1}}\right)
\end{equation}
Next, we consider two different cases, first when there are odd number of nodes in the network and second when there are even number of nodes in the network.\\
\textbf{Odd Number of Nodes:} When there are odd number of nodes in the network, the majority fusion rule is $K^*=(N+1)/2$. In this case \eqref{majcon1} is equivalent to show that 
\begin{equation}
\left(\dfrac{N-1}{2}\right)\left(\dfrac{1-P_f}{\pi_{1,0}}-\dfrac{1-P_d}{\pi_{1,1}}\right)>\left(\dfrac{N-1}{2}\right)\left(\dfrac{1-P_f}{1-\pi_{1,0}}-\dfrac{1-P_d}{1-\pi_{1,1}}\right). \label{evnmajo1}
\end{equation}
To show that the above condition is true, we start from the following inequality 
\begin{eqnarray*}
&&
\dfrac{(1-\pi_{1,0})(1-\pi_{1,1})}{\pi_{1,0}\pi_{1,1}}>-1\\
&\Leftrightarrow&
\left[\dfrac{1}{\pi_{1,0}}-\dfrac{1}{\pi_{1,1}}\right]>\left[\dfrac{1}{1-\pi_{1,0}}-\dfrac{1}{1-\pi_{1,1}}\right]\\
&\Leftrightarrow&
\left[\dfrac{1}{\pi_{1,0}}-\dfrac{1}{1-\pi_{1,0}}\right]>\left[\dfrac{1}{\pi_{1,1}}-\dfrac{1}{1-\pi_{1,1}}\right]
\end{eqnarray*}
Since $\dfrac{1-P_f}{1-P_d}>1$, $\pi_{1,0}<0.5$ (consequence of our assumption) and $N\geq 2$, the above condition is equivalent to 
\begin{eqnarray}
&&
\dfrac{1-P_f}{1-P_d}\left[\dfrac{1}{\pi_{1,0}}-\dfrac{1}{1-\pi_{1,0}}\right]>\left[\dfrac{1}{\pi_{1,1}}-\dfrac{1}{1-\pi_{1,1}}\right]\nonumber \\
&\Leftrightarrow&
\left(\dfrac{1-P_f}{\pi_{1,0}}-\dfrac{1-P_d}{\pi_{1,1}}\right)>\left(\dfrac{1-P_f}{1-\pi_{1,0}}-\dfrac{1-P_d}{1-\pi_{1,1}}\right)\nonumber \\
&\Leftrightarrow&
\left(\dfrac{N-1}{2}\right)\left(\dfrac{1-P_f}{\pi_{1,0}}-\dfrac{1-P_d}{\pi_{1,1}}\right)>\left(\dfrac{N-1}{2}\right)\left(\dfrac{1-P_f}{1-\pi_{1,0}}-\dfrac{1-P_d}{1-\pi_{1,1}}\right)  
\end{eqnarray}
which implies that $\dfrac{dr\left( P_{1,0}, K^*, \alpha \right)}{d P_{1,0}}>0$ for odd number of nodes case. Next, we consider the even number of nodes case.\\
\textbf{Even Number of Nodes:} Now, we consider the case when there are even number of nodes in the network and majority fusion rule is given by $K^*=\dfrac{N}{2}+1$. Condition \eqref{majcon1} is equivalent to show that
 \begin{eqnarray*}
   &&
\left(\dfrac{N}{2}\right)\left(\dfrac{1-P_f}{\pi_{1,0}}-\dfrac{1-P_d}{\pi_{1,1}}\right)>\left(\dfrac{N}{2}-1\right)\left(\dfrac{1-P_f}{1-\pi_{1,0}}-\dfrac{1-P_d}{1-\pi_{1,1}}\right).  \end{eqnarray*}
 Which follows from the fact that 
 \begin{equation*}
 \left(\dfrac{N}{2}\right)\left(\dfrac{1-P_f}{\pi_{1,0}}-\dfrac{1-P_d}{\pi_{1,1}}\right)>\left(\dfrac{N}{2}-1\right)\left(\dfrac{1-P_f}{\pi_{1,0}}-\dfrac{1-P_d}{\pi_{1,1}}\right)
  \end{equation*} 
 and the result given in \eqref{evnmajo1}. This completes our proof.

\section{Calculating partial derivative of $P_E$ w.r.t. $P_{1,0}$}
\label{proof3} 
First, we calculate the partial derivative of $Q_F$ with respect to $P_{1,0}$. Notice that,
\begin{equation}
\label{derqf}
     Q_{F} = \sum_{i = K^*}^{N} \nchoosek{N}{i} (\pi_{1,0})^i (1-\pi_{1,0})^{N-i}
     \end{equation}
where
\begin{eqnarray}
\pi_{1,0}&=&\alpha(P_{1,0}(1-P_f)+(1-P_{0,1})P_f)+(1-\alpha)P_f\\\label{conddist}
(\pi_{1,0})'&=& d\pi_{1,0}/dP_{1,0} =\alpha(1-P_f).
\end{eqnarray}  
Differentiating both sides of \eqref{derqf} with respect to $P_{1,0}$, we get           
\begin{eqnarray*}
\dfrac{d Q_F}{d P_{1,0}}&=& \nchoosek{N}{K^*}(K^*(\pi_{1,0})^{K^*-1}(\pi_{1,0})'(1-\pi_{1,0})^{N-K^*}-(\pi_{1,0})^{K^*}(N-K^*)(1-\pi_{1,0})^{N-K^*-1}(\pi_{1,0})')\\      &+&\nchoosek{N}{K^*+1}((K^*+1)(\pi_{1,0})^{K^*}(\pi_{1,0})'(1-\pi_{1,0})^{N-K^*-1}-(\pi_{1,0})^{K^*+1}(N-K^*-1)\\
&(& 1-\pi_{1,0})^{N-K^*-2}(\pi_{1,0})') + \cdots +\nchoosek{N}{N}(N(\pi_{1,0})^{N-1}(\pi_{1,0})'-0)\\           
&=&(\pi_{1,0})'(\pi_{1,0})^{K^*-1}(1-\pi_{1,0})^{N-K^*}\Bigg[\nchoosek{N}{K^*}\left(K^*-\dfrac{\pi_{1,0}}{1-\pi_{1,0}}(N-K^*)\right)\\
 &+&\nchoosek{N}{K^*+1}\left((K^*+1)\dfrac{\pi_{1,0}}{1-\pi_{1,0}}-(N-K^*-1)\left(\dfrac{\pi_{1,0}}{1-\pi_{1,0}}\right)^2\right)+\cdots\Bigg]\\ 
&=&(\pi_{1,0})'(\pi_{1,0})^{K^*-1}(1-\pi_{1,0})^{N-K^*}\Bigg[\nchoosek{N}{K^*}(K^*-\dfrac{\pi_{1,0}}{1-\pi_{1,0}}(N-K^*))\\
&+&\dfrac{\pi_{1,0}}{1-\pi_{1,0}}\nchoosek{N}{K^*+1}\left((K^*+1) -(N-K^*-1)\dfrac{\pi_{1,0}}{1-\pi_{1,0}}\right)+\cdots\Bigg]\\
&=&(\pi_{1,0})'(\pi_{1,0})^{K^*-1}(1-\pi_{1,0})^{N-K^*}\Bigg[\nchoosek{N}{K^*}K^*+\Bigg[-\dfrac{\pi_{1,0}}{1-\pi_{1,0}}\nchoosek{N}{K^*}(N-K^*)\\
\end{eqnarray*}
\begin{equation*}
+\dfrac{\pi_{1,0}}{1-\pi_{1,0}}\nchoosek{N}{K^*+1}(K^*+1)\Bigg]+\cdots\Bigg]
\end{equation*}
Since, $\nchoosek{N}{K^*}\dfrac{K^*}{N}=\nchoosek{N-1}{K^*-1}$, the above equation can be written as
\begin{eqnarray}
\dfrac{d Q_F}{d P_{1,0}}&=&
 (\pi_{1,0})'(\pi_{1,0})^{K^*-1}(1-\pi_{1,0})^{N-K^*} \Bigg[\nchoosek{N-1}{K^*-1}N \nonumber \\
  &+&\dfrac{\pi_{1,0}}{1-\pi_{1,0}}\Bigg\{\nchoosek{N}{K^*+1}(K^*+1)-\nchoosek{N}{K^*}(N-K^*)\Bigg\}+\cdots\Bigg].\label{mech2}
\end{eqnarray}
 Notice that, for any positive integer $t$
 \begin{equation}
 \label{mech1}
\left(\dfrac{\pi_{1,0}}{1-\pi_{1,0}}\right)^{t}\left[\nchoosek{N}{K^*+t}(K^*+t)-\nchoosek{N}{K^*+t-1}(N-K^*-t+1)\right]=0. 
 \end{equation}
Using the result from \eqref{mech1}, \eqref{mech2} can be written as
 \begin{eqnarray*}
 &&
  \dfrac{dQ_F}{dP_{1,0}}=(\pi_{1,0})'(\pi_{1,0})^{K^*-1}(1-\pi_{1,0})^{N-K^*}\left[\nchoosek{N-1}{K^*-1}N+\dfrac{\pi_{1,0}}{1-\pi_{1,0}}[0]+\cdots+[0]\right]\\
&\Leftrightarrow&
 \dfrac{d Q_F}{d P_{1,0}}= \alpha(1-P_f) N \nchoosek{N-1}{K^*-1} \left( \pi_{1,0} \right)^{K^*-1} \left( 1-\pi_{1,0} \right)^{N-K^*}  .       
 \end{eqnarray*}
 
Similarly, the partial derivative of $Q_D$ w.r.t. $P_{1,0}$ can calculated to be
\begin{equation*}
\dfrac{d Q_D}{d P_{1,0}}= \alpha(1-P_d) N \nchoosek{N-1}{K^*-1} \left( \pi_{1,1} \right)^{K^*-1} \left( 1-\pi_{1,1} \right)^{N-K^*}.      
\end{equation*}

\bibliographystyle{IEEEtran}
\bibliography{Conf,Book,Journal}

\end{document}